\documentclass[a4,11pt]{article}
\usepackage{amsmath, amssymb}
\usepackage{latexsym}
\usepackage{amsthm}
\usepackage{color}

\makeatletter
\@addtoreset{equation}{section}
\makeatother

\newtheorem{Theorem}{Theorem}[section]
\newtheorem{lemma}[Theorem]{Lemma}
\newcommand{\hz}{\frac{1}{H_0}}

\def\qed{\hfill $\Box$} 

\title{\huge Mass Renormalization in the Nelson Model }
\author{\textbf{Susumu Osawa}\\
Department of Mathematics, Hokkaido University\\
Sapporo, Hokkaido 060-0810, Japan}
 \author{Fumio Hiroshima\thanks{
Faculty of Mathematics, Kyushu Univesity,  
Fukuoka 819-0385, Japan, 
e-mail:hiroshima@math.kyushu-u.ac.jp}
and 
Susumu Osawa\thanks{
Faculty of Science, Department of Mathematics, Hokkaido University,
Sapporo, Hokkaido 060-0810, Japan
e-mail:
osawa@math.sci.hokudai.ac.jp}}
\date{}

\begin{document}
\maketitle
\setlength{\baselineskip}{12pt}

\begin{abstract}
The asymptotic behavior of the effective mass $m_{\rm eff}(\Lambda)$  of the so-called 
Nelson model in  quantum field theory is considered, where $\Lambda$ is an ultraviolet cutoff  parameter of the model. 
Let $m$ be the bare mass of the model.
It is shown that for sufficiently small coupling constant $|\alpha|$ of the model, 
$m_{{\rm eff}}(\Lambda)/m$ can be expanded as 
$m_{{\rm eff}}(\Lambda)/m= 1+\sum_{n=1}^\infty a_n(\Lambda) \alpha^{2n}$.  
A physical folklore  is  that 
$a_n(\Lambda)=O( [\log \Lambda]^{(n-1)})$ as $\Lambda\to \infty$. 
It is rigorously shown that 
$$0<\lim_{\Lambda\to\infty}a_1(\Lambda)<C,\quad   C_1\leq \lim_{\Lambda\to\infty}a_2(\Lambda)/\log\Lambda\leq C_2$$
with some  constants $C$, $C_1,$ and $C_2$.
\end{abstract}


\section{Introduction and main results}
The model considered in this paper is the so-called Nelson model~\cite{Nelson}, which describes a nonrelativistic nucleon with bare mass $m>0$ interacting with a quantized scalar field with mass $\nu>0$.  
The nucleon is governed by a Schr\"odinger operator. 
Let us first define the Nelson Hamiltonian. We use relativistic unit and employ the total momentum representation. Then the Hilbert space of states is the boson Fock space over $L^2(\mathbb{R}^3)$ which is given by 
$$  \mathcal{F} = \oplus_{n=0}^{\infty}[\otimes_s^n L^2(\mathbb{R}^3)], $$ 
where $\otimes_s^n$ denotes the $n$-fold symmetric tensor product and $\otimes_s^0 L^2(\mathbb{R}^3) = \mathbb{C}.$ Then $\Phi\in\mathcal{F}$ can be written as $\Phi=\{ \Phi^{(0)},  \Phi^{(1)}, \Phi^{(2)}, \dots \}$, where $\Phi^{(n)}\in\otimes_s^n L^2(\mathbb{R}^3)$.
The Fock vacuum $\Omega \in \mathcal{F}$ 
is defined by $\Omega = \{1, 0, 0, \dots \}$. Let $a(f),\; f \in L^2(\mathbb{R}^3)$, be   the annihilation operator and $a(f)^{\ast},\; f \in L^2(\mathbb{R}^3)$, the creation operator on $\mathcal{F}$, which are defined by 
\begin{align*}
& D(a(f)^{\ast})=\{ \Psi \in \mathcal{F}| \sum_{n=0}^{\infty}(n+1)\| S_{n+1}(f \otimes \Psi^{(n)})\| ^2_{\otimes^n L^2(\mathbb{R}^3)}<\infty \}, \\
& (a(f)^{\ast}\Psi)^{(0)} =0, \\
& (a(f)^{\ast}\Psi)^{(n+1)} = \sqrt{n+1}S_{n+1}(f \otimes \Psi^{(n)}), 
\end{align*} 
and $a(f)=(a(f)^{\ast})^{\ast}$, where $S_n$ is the symmetrizer, $D(X)$ the domain of operator $X$, and $\| \cdot\| _{\mathcal{K}}$ the norm on $\mathcal{K}$. They satisfy canonical commutation relations as follows:
$$  [ a(f), a(g)^{\ast}]=(f, g), \quad [a(f), a(g)]=0, \quad   [a(f)^{\ast}, a(g)^{\ast}]=0 $$ 
on a suitable dense domain, where $[X, Y]=XY-YX$ and $(\cdot , \cdot )$ is the inner product on $\mathcal{K}$ (linear in the second variable). Let $T$ be a self-adjoint operator on  $L^2(\mathbb{R}^3)$. Then we define the self-adjoint operator ${\rm d}\Gamma(T)$ on $\mathcal{F}$ by
$ {\rm d}\Gamma(T)= \oplus_{n=0}^{\infty}T^{(n)}$, 
where 
\begin{align*}
T^{(n)}=\overline{\left(\sum_{j=1}^{n}I\otimes\cdots\otimes I\otimes \stackrel{j\rm{th}}{\breve{T}}\otimes I\otimes \cdots \otimes I\right)\lceil{\stackrel{n}{\hat\otimes} D(T)}}
\qquad (n \ge 1)
\end{align*}
with 
$T^{(0)}=0$.  
Here, for a closable operator $T,$ $\overline{T}$ denotes the closure of $T.$ The operator ${\rm d}\Gamma(T)$ is called the second quantization of $T$. The free energy of the scalar field is given by $H_{\rm f}= {\rm d}\Gamma(\omega)$, where $\omega(k)=\sqrt{|k|^2+\nu ^2}$ $(k=(k_1, k_2, k_3) \in \mathbb{R}^3 , \;  \nu > 0)$ is considered as a multiplication operator on $L^2(\mathbb{R}^3)$. Similarly the momentum of the scalar field is given by $P_{{\rm f}\mu}= {\rm d}\Gamma (k_{\mu})\; (\mu=1, 2, 3)$. The coupling of the nucleon and a scalar field is mediated through the Segal field operator $\Phi_s(g)$ defined by
 $$  \Phi_s(g)=\frac{1}{\sqrt{2}}(a(g)+a(g)^{\ast}), $$ 
where $g$ is a cutoff function given by 
$
g(k)=\frac{\hat{\phi}(k)}{\sqrt{\omega(k)}}$. 
Here $\hat{\phi}$ is the form factor with infrared cutoff $\kappa>0$ and ultraviolet cutoff $\Lambda>0$, which are defined by
\begin{align*}
  \hat{\phi}(k)= \left\{ \begin{array}{ll}
    0   &  |k| <   \kappa,  \\
    (2\pi)^{-\frac{3}{2}} & \kappa \le |k| \le \Lambda, \\
    0   & |k| >   \Lambda.
  \end{array} \right. 
\end{align*}
The Nelson Hamiltonian with total momentum $p\in\mathbb{R}^3$  is given by a self-adjoint operator on $\mathcal{F}$ as follows:
$$ H(p)=\frac{1}{2m}(p-P_{\rm f})^2+H_{\rm f}+\alpha \Phi_s(g), $$ 
where $\alpha \in \mathbb{R}$ is a coupling constant. Let $E(p, \alpha)$ be the energy-momentum relation (the infimum of the spectrum  $\sigma (H(p))$) defined by
\[ E(p, \alpha)= \inf\sigma(H(p)).  \]
Then the effective mass $m_{\rm eff}=m_{\rm eff}(\Lambda)$ is defined by 
$$\frac{1}{m_{\rm eff}}=\frac{1}{3}\Delta_{p}E(p, \alpha)\upharpoonright_{p=0} .$$ 
Here $\Delta_{p}$ denotes the three dimensional Laplacian in the variable $p.$
We are concerned with the asymptotic behavior of $m_{\rm eff}$ as the ultraviolet cutoff goes to infinity. It is however a subtle problem. Removal of the ultraviolet cutoff $\Lambda$ through mass renormalization means finding sequences $\{m\}$ and $\{\Lambda \}$ such that $m \to 0$, $ \Lambda\to\infty$, and $m_{\rm eff}$ converges.
Since we can see that $m_{\rm eff}/m$ is a function of $\Lambda/m$, 
to achieve this, we want to find constants $0<\gamma <1$ and $0<b_0<\infty$ such that 
\begin{align}
\lim_{\Lambda\to\infty}\frac{m_{\rm eff}/m}{(\Lambda/m)^{\gamma}}=b_0. \label{eq:asymptoticmeff}
\end{align}
If we succeed in finding constants $\gamma$ and $b_0$ such as in (\ref{eq:asymptoticmeff}), scaling the bare mass $m$ as 
$$m=\frac{1}{\Lambda^{\gamma/(1-\gamma)}}M,$$ 
where $M=(\frac{m^{\ast}}{b_0})^{1/(1-\gamma)}$ with an arbitrary positive  constant $m^\ast$, 
we have 
\begin{align*}
\lim_{\Lambda\to\infty}m_{\rm eff}(\Lambda)=m^{\ast}.
\end{align*}
The mass renormalization is, however, a subtle problem, and unfortunately, we cannot yet find constants $\gamma$ and $b_0$ such as in (\ref{eq:asymptoticmeff}). For that reason we turn to perturbartive renormalization, by which we try to guess the proper value of $\gamma.$
Main results obtained in this paper are summarized as follows.

\noindent
{\bf Theorem \ref{1}}\ 
Let $\kappa>0$. Then $m_{\rm eff}$ is an analytic function of $\alpha^2$ and can be expanded in the following power series for sufficiently small $|\alpha|$:
$$  \frac{m_{\rm eff}}{m}=1+\sum_{n=1}^{\infty}a_n(\Lambda)\alpha ^{2n}. $$  

\noindent
{\bf Theorem \ref{2}}\
There exists a strictly positive constant $C$ such that
$$ \lim_{\Lambda\to\infty}a_1(\Lambda)=C.$$

\noindent
{\bf Theorem \ref{3}}\
There exist some constants $C_1$ and $C_2$ such that
$$ 
C_1 \le \lim_{\Lambda\to\infty}\frac{a_2(\Lambda)}{\log\Lambda}\le C_2. 
$$ 

From Theorems \ref{2} and \ref{3}, if 
$\displaystyle D= \lim_{\Lambda\to\infty}a_2(\Lambda)/\log\Lambda>0,$ it is suggested that $\gamma=D\alpha^2/C.$ So, the mass of  the Nelson model is renormalizeble  for sufficiently small $|\alpha|$.

The effective mass and energy-momentum relation have been studied mainly in nonrelativistic electrodynamics.  Spohn~\cite{Spohn1} investigates the upper and lower bound of the effective mass of the polaron model from a functional integral point of view. Hiroshima and Spohn~\cite{HiroshimaSpohn} study a perturbative mass renormalization including fourth order in the coupling constant in the case of a spinless electron. Hiroshima and Ito~\cite{HiroshimaIto1, HiroshimaIto2} study it in the case of an electron with spin $1/2$.  Bach, Chen, Fr\"ohlich, and Sigal~\cite{Bach} show that the energy-momentum relation is equal to the infimum of the essential spectrum of the Hamiltonian for $\kappa \ge 0.$ Fr\"ohlich and Pizzo~\cite{FrohlichPizzo} investigate energy-momentum relation when infrared cutoff goes to~0.

\section{Analytic properties}
In order to investigate the effective mass in a perturbation theory we have to 
check the analytic properties of $E(p,\alpha)$. 
\subsection{Analytic family in the sense of Kato}
\begin{lemma}\label{Katofamily}
$H(p)$ is an analytic family in the sense of Kato.
\end{lemma}
\begin{proof}
We prove $H(p)$ is an analytic family of type (A).
We see that 
\[
H(p)=H_0+\sum_{\mu=1}^3p_{\mu}\frac{1}{2m}(p_{\mu}-2P_{{\rm f}\mu})+\alpha H_{\rm I},
\]
where 
$ H_0=\frac{1}{2m}P_{\rm f}^2+H_{\rm f}$ and $H_{\rm I}= \Phi_s(g)$. 
Hence all we have to do is to prove the following facts.
\begin{enumerate}
\item[(a)]
$D(H_0)\subset \cap_{\mu=1}^3D(P_{{\rm f}\mu})\cap D(H_{\rm I}).$
\item[(b)]
There exist real constants $a_{\mu},$ $b_{\mu}$ $(\mu=1, 2, 3),$ $c,$ and $d$ such that for any $\Psi \in D(H_0)$
\begin{align*}
&\| \frac{1}{2m}(p_{\mu}-2P_{{\rm f}\mu})\Psi\| _{\mathcal{F}}\le a_{\mu}\| H_0\Psi\| _{\mathcal{F}}+b_{\mu}\| \Psi\| _{\mathcal{F}} \quad (\mu=1, 2, 3),\\
&\| H_{\rm I}\Psi\| _{\mathcal{F}}\le c\| H_0\Psi\| _{\mathcal{F}}+d\| \Psi\| _{\mathcal{F}}. \end{align*}
\end{enumerate}
We prove (a) at first. 
Since $\cap_{\mu=1}^3D(P_{{\rm f}\mu}^2)\subset D(P_{{\rm f}\mu})$, we have 
$D(H_0)=\cap_{\mu=1}^3D(P_{{\rm f}\mu}^2)\cap D(H_{\rm f})\subset\cap_{\mu=1}^3D(P_{{\rm f}\mu})$. 
Additionally, since
$ \| \omega^{-1/2}g\| _{L^2(\mathbb{R}^3)}
<\infty$, 
we have $g\in D(\omega^{-1/2}).$ Furthermore, since $\omega$ is a nonnegative and injective self-adjoint operator on $L^2({\mathbb R}^3)$, it follows that 
$D({\rm d}\Gamma(\omega)^{1/2})\subset D(a(g))\cap D(a(g)^{\ast})=D(H_{\rm I})$.  
Hence we have
$
D(H_0)\subset D({\rm d}\Gamma(\omega))\subset D({\rm d}\Gamma(\omega)^{1/2})$. 
Together with them, (a) is proven . Next we prove  (b).
Let $\Psi$ be an arbitrary vector in $D(H_0).$ Then we have
$$ \| \frac{1}{2m}(p_{\mu}-2P_{{\rm f}\mu})\Psi\| _{\mathcal{F}}\le \frac{|p_{\mu}|}{2m}\| \Psi\| _{\mathcal{F}}+\frac{1}{m}\| P_{{\rm f}\mu}\Psi\| _{\mathcal{F}}. $$ 
Since 
$
 \| P_{{\rm f}\mu}\Psi\| _{\mathcal{F}}^2\le2m\| H_0^{1/2}\Psi\| _{\mathcal{F}}^2$, 
we have 
$ \| H_0^{1/2}\Psi\| _{\mathcal{F}}^2\le\| (H_0+1)\Psi\| _{\mathcal{F}}^2$. 
Hence
\begin{align}
\| \frac{1}{2m}(p_{\mu}-2P_{{\rm f}\mu})\Psi\| _{\mathcal{F}}&\le \sqrt{\frac{2}{m}}\| H_0\Psi\| _{\mathcal{F}}+(\frac{|p_{\mu}|}{2m}+\sqrt{\frac{2}{m}})\| \Psi\| _{\mathcal{F}}.   \label{typeA1}
\end{align}
Since $D(H_0)\subset D({\rm d}\Gamma(\omega)^{1/2}),$ 
\begin{align*}
&\| a(g)\Psi\| _{\mathcal{F}}\le\| \omega^{-1/2}g\| _{L^2(\mathbb{R}^3)}
\| H_{\rm f}^{1/2}\Psi\| _{\mathcal{F}},\\
&
\| a(g)^{\ast}\Psi\| _{\mathcal{F}}\le\| \omega^{-1/2}g\| _{L^2(\mathbb{R}^3)}\| H_{\rm f}^{1/2}\Psi\| _{\mathcal{F}}+\| g\| _{L^2(\mathbb{R}^3)}\| \Psi\| _{\mathcal{F}} 
\end{align*}
hold. Hence
\begin{align*}
 \| H_{\rm I}\Psi\| _{\mathcal{F}}
 \le\sqrt{2}\| \omega^{-1/2}g\| _{L^2(\mathbb{R}^3)}\| H_{\rm f}^{1/2}\Psi\| _{\mathcal{F}}+\frac{1}{\sqrt{2}}\| g\| _{L^2(\mathbb{R}^3)}\| \Psi\| _{\mathcal{F}}.
\end{align*}
From 
triangle inequality, we have
$  \| H_{\rm f}^{1/2}\Psi\| _{\mathcal{F}}\le\| H_{\rm f}\Psi\| _{\mathcal{F}}+\| \Psi\| _{\mathcal{F}}$. 
In addition, 
\[
\| H_0\Psi\| _{\mathcal{F}}^2-\| H_{\rm f}\Psi\| _{\mathcal{F}}^2=\| \frac{1}{2m}P_{\rm f}^2\Psi\| _{\mathcal{F}}^2+\frac{1}{m}\Re(P_{\rm f}^2\Psi, H_{\rm f}\Psi). 
\] 
Since $P_{\rm f}^2$ and $H_{\rm  f}$ are strongly commutative and nonnegative self-adjoint operators on $\mathcal{F}$, $(P_{\rm f}^2\Psi, H_{\rm f}\Psi)\ge 0$ holds. Hence $\| H_{\rm f}\Psi\| _{\mathcal{F}}\le\| H_0\Psi\| _{\mathcal{F}}.$
Then we have
\begin{align}
\| H_{\rm I}\Psi\| _{\mathcal{F}}&\le\sqrt{2}\| \omega^{-1/2}g\| _{L^2(\mathbb{R}^3)}\| H_0\Psi\| _{\mathcal{F}}+(\sqrt{2}\| \omega^{-1/2}g\| _{L^2(\mathbb{R}^3)}+\frac{1}{\sqrt{2}}\| g\| _{L^2(\mathbb{R}^3)})\| \Psi\| _{\mathcal{F}}. \label{typeA2}
\end{align}
From (\ref{typeA1}) and (\ref{typeA2}), (b) is proven . Hence $H(p)$ is an analytic family of  type (A).  Since every analytic family of type (A) is an analytic family of in the sense of Kato, it is an analytic family in the sense of Kato.
\end{proof}
We denote the ground state of $H(p)$ by $\psi_g(p).$

\begin{lemma}
(1) 
$E(p,\alpha)$ is analytic in $p$ and $\alpha$ if $|p|$ and $|\alpha|$ are sufficiently small.
(2)
$\psi_g(p)$ is strongly analytic in $p$ and $\alpha$ if $|p|$ and $|\alpha|$ are sufficiently small.
\end{lemma}
\begin{proof}
From \cite[Theoren XII.9]{ReedSimon}, (1) follows, and 
from \cite[Theoren XII.8]{ReedSimon}, (2) follows.
\end{proof}

\subsection{Formula}
In this section we expand $m/m_{\rm eff}$ with respect to 
$\alpha$. 
\begin{lemma} \label{formula}
The ratio $m/m_{\rm eff}$ can be expressed as
\begin{align*}
\frac{m}{m_{\rm eff}}=1-\frac{2}{3}\sum_{\mu=1}^3\frac{(P_{{\rm f}\mu}\psi_g(0), {\psi_g}_{\mu}'(0))}{(\psi_g(0), \psi_g(0))},  
\end{align*}
where ${\psi_g}_{\mu}'(0)=s-{\partial p_{\mu}}\psi_g(p)\upharpoonright_{p=0}.$
\end{lemma}
\begin{proof}
Since $E(p, \alpha)$ is symmetry, $E(p, \alpha)=E(-p, \alpha),$ we have
\begin{align}
\partial p_{\mu}E(p, \alpha)\upharpoonright_{p=0}=0, \quad \mu=1, 2, 3. \label{eq:Ederivative}
\end{align}
Since 
$   H(p)\psi_{g}(p)=E(p,\alpha)\psi_{g}(p)$, 
for any $\Psi\in D(H(p)),$
\begin{align*}
(H(p)\Psi, \psi_{g}(p))=E(p,\alpha)(\Psi,\psi_{g}(p))
\end{align*}
holds. Taking a derivative with respect to $p_{\mu}$  on the both sides above, we have
\begin{align}
&(H_{\mu}'(p)\Psi, \psi_g(p))+(H(p)\Psi, \psi_{g_{\mu}}'(p))=E'_{\mu}(p, \alpha)(\Psi, \psi_g(p))+E(p, \alpha)(\Psi, \psi_{g_{\mu}}'(p)) \label{pmuderivative1},\\
&(H_{\mu}^{\prime\prime}(p)\Psi,\psi_{g}(p))
+2(H_{\mu}'(p)\Psi,\psi_{g_{\mu}}'(p))
+(H(p)\Psi,\psi_{g_{\mu}}^{\prime\prime}(p)) \notag \nonumber\\
&=E_{\mu}^{\prime\prime}(p,\alpha)(\Psi, \psi_g(p))+2E_{\mu}'(p,\alpha)(\Psi,\psi_{g_{\mu}}'(p))
+E(p,\alpha)(\Psi,\psi_{g_{\mu}}\prime\prime(p)).\nonumber
\end{align}
Here ${\ }'$ denotes the derivative or strong derivative with respect to $p_{\mu}$, and  $H_{\mu}'(p)=\frac{1}{m}(p_{\mu}-P_{{\rm f}\mu}),$ $H_{\mu}^{\prime\prime}(p)=\frac{1}{m}.$
Setting $\Psi=\psi_g(0)$ and $p=0,$ we have
$$ 
E_{\mu}^{\prime\prime}(0,\alpha)=\frac{1}{m}\frac{(\psi_g(0), \psi_g(0))-2(P_{{\rm f}\mu}\psi_g(0),\psi_{g{\mu}}\prime(0))}{(\psi_g(0), \psi_g(0))}.
$$ 
This expression and the definition of the effective mass prove the lemma.

\end{proof}

\subsection{Perturbative expansions}
We define operators $A^{+}$ and $A^{-}$ by
$A^{+} = \frac{1}{\sqrt{2}}a(g)^{\ast}$ and $A^{-}=\frac{1}{\sqrt{2}}a(g)$. 
Then
$ H_{\rm I}=A^{+}+A^{-}$. 
Moreover, let 
$ 
{\mathcal F}^{(n)}= \otimes_s^n L^2(\mathbb{R}^3)$ 
and  
\begin{align}
\psi_g(0)=\sum_{n=0}^{\infty}\frac{\alpha ^n}{n!}\varphi_n. \label{eq:tenkai1}
\end{align}
Since $E(p, \alpha)$ is symmetry $E(p, -\alpha)=E(p, \alpha)$, we have
\begin{align} 
E(0,\alpha)=\sum_{n=0}^{\infty}\frac{\alpha^{2n}}{(2n)!}E_{2n}. \label{eq:tenkai2}
\end{align}
Since $\ker H_0\neq \{0\}$, $H_0$ is not injective. 
However, we define the operator $\hz $(for notational simplicity we write $\hz$ for $H_0^{-1}$ in what follows) on ${\mathcal F}$ as follows.
\begin{align*}
&D(\hz )=\left\{ 
\Psi=\oplus_{n=0}^\infty\Psi^{(n)}\in{\mathcal F} |\  \sum_{n=1}^{\infty}\left\| 
\beta^n 
\Psi^{(n)}\right\| ^2<\infty \right\}, \\
&(\hz \Psi)^{(0)}=0, \\   
&(\hz \Psi)^{(n)}(k_1,\cdots ,k_n)=\beta^n(k_1,\cdots,k_n)
\Psi^{(n)}(k_1,\cdots ,k_n)  \quad (n\ge1).
\end{align*}
Here 
$$\beta^n=\beta^n(k_1,\cdots,k_n)=
\frac{1}{\frac{1}{2m}|k_1+\cdots + k_n|^2+\sum_{i=1}^n\omega(k_i)}.
$$
We define the subspace ${\mathcal F}_{\rm fin}$ of ${\mathcal F}$ as
${\mathcal F}_{\rm fin}=\{ \{\Psi^{(n)}\}_{n=0}^{\infty}\in {\mathcal F}  |\Psi^{(l) }=0\; \text{for}\: l\ge q \;\text{with some}\; q  \}$. 

\begin{lemma} \label{FfinDH0inverse}
It holds that ${\mathcal F}_{\rm fin}\subset D(\hz )$.
\end{lemma}
\begin{proof}
Let
$\Psi\in{\mathcal F}_{\rm fin}.$ Then
$
 \| \hz \Psi\| ^2=\sum_{n=1}^{\infty}\| (\hz \Psi)^{(n)}\| ^2 <\infty$.  
 Hence the lemma follows.

\end{proof}

\begin{lemma} \label{lemma;recurrence1}
Let $\psi_g(0)=\sum_{n=0}^{\infty}(\alpha^n/n!)\varphi_n$. Then
\label{lemma:varphi}
$\varphi_0=\Omega,$ $\varphi_1=-\hz  H_{\rm I} \Omega$ and the recurrence formulas
\begin{align}
\varphi_{2l}&=\hz \{-2lH_{\rm I} \varphi _{2l-1} 
+\sum_{j=1}^{l}{2l \choose 2j} E_{2j} \varphi _{2l-2j}\} \quad(l \ge 1), \label{phi2l}\\
\varphi_{2l+1}&=\hz \{-(2l+1)H_{\rm I} \varphi _{2l} +\sum_{j=1}^{l} {2l+1 \choose 2j}E_{2j} \varphi _{2l+1-2j}\} \quad (l \ge 0) \label{phi2lplus1}
\end{align}
follow, with
\begin{align}
&\varphi_{2l}\in {\mathcal F}^{(2)} \oplus {\mathcal F}^{(4)} \oplus \cdots \oplus {\mathcal F}^{(2l)} \quad (l \ge 1), \label{sumphi2l}  \\
&\varphi_{2l+1}\in {\mathcal F}^{(1)} \oplus {\mathcal F}^{(3)} \oplus \cdots \oplus {\mathcal F}^{(2l+1)} \quad (l \ge 0),
\label{sumphi2l+1}
\end{align}
and $E_{2l}$ is given by 
\begin{align}
E_{2l}=2l(\Omega, H_{\rm I} \varphi_{2l-1}) \quad (l \ge 1). \label{E2l}
\end{align}
\end{lemma}
\begin{proof}
We have $E(0, 0)=E_0$ by substituting $\alpha=0$ in (\ref{eq:tenkai2}).  Since $E(0, 0)$ is the ground state energy of $H_0$, $E(0, 0)=0.$ Hence $E_0=0.$ Since $\varphi_0$ is the ground state of $H_0$, $\varphi_0$ can be $\Omega$. 
We can find that 
$
(\varphi_n ,\Omega)=\delta_{0 n}$ for $n=0,1, \cdots$
holds in the same way as~\cite{HiroshimaSpohn}. From now we set $H=H(p),$ $\psi_g=\psi_g(0),$ $E=E(p,\alpha)$ and $\prime$ means (strong)derivative with respect to $\alpha$.
\begin{align}
(H\Psi, \psi_g)=E(\Psi, \psi_g) \label{naiseki2}
\end{align}
holds for $\Psi\in D(H)$. Differentiating (\ref{naiseki2}) with respect to $\alpha$, we have
$$  (H_{\rm I} \Psi, \psi_g)+ (H\Psi, \psi _g ')=E'(\Psi, \psi_g)+E(\Psi, \psi_g'). $$ 
Hence  $\psi _g '\in D(H)$ and we have
\begin{align}
 H_{\rm I}\psi_g+ H\psi_g' = E'\psi_g+E\psi_g'. \label{alphaderibative}
\end{align}
Substituting $p=0$ and $\alpha=0$ into (\ref{alphaderibative}) and taking into account     $(\varphi_n ,\Omega)=\delta_{0 n}$, we have $\varphi_1=-\hz  H_{\rm I} \Omega$. Differentiating (\ref{naiseki2}) $n$ times with respect to $\alpha$, we also have
$$  (H\Psi, \psi_g^{(n)})+n(\Psi, H_{\rm I}\psi_g^{(n-1)})=\sum_{j=1}^n {n \choose j}E^{(j)}(\Psi, \psi_g^{(n-j)}). $$ 
By the induction on $n$, we have $\psi_g^{(n)} \in D(H)$ and
$$  H\psi_g^{(n)}+nH_{\rm I}\psi_g^{(n-1)}=\sum_{j=1}^n {n \choose j}E^{(j)}\psi_g^{(n-j)}. $$ 
Substituting $p=0$ and $\alpha=0$ into both sides above, we have
\begin{align*}
&H_0 \varphi_{2l}+ 2l H_{\rm I} \varphi_{2l-1} =\sum_{j=1}^l{2l \choose 2j}E_{2j}\varphi_{2l-2j} \quad (l \ge 1), \\ 
&H_0 \varphi_{2l+1}+ (2l+1) H_{\rm I} \varphi_{2l} =\sum_{j=1}^l{2l+1 \choose 2j}E_{2j}\varphi_{2l+1-2j} \quad (l \ge 0). 
\end{align*} 
From now on, we shall prove 
\begin{align*}
&\varphi_n^{(i)}=0, \quad (i>n, \; i=0 ),\\
&\varphi_{2l}^{(2i+1)}=0,\quad (l \ge 1, \;0\le i\le l-1),\\
&{\rm supp}_{k\in \mathbb{R}^{3\cdot2i}}\varphi_{2l}^{(2i)}(k)=S_{2i} \mbox{ or }
\emptyset,\quad (l\ge1,\; 1\le i\le l),\\
&\varphi_{2l+1}^{(2i)}=0,\quad (l\ge0,\; 0\le i \le l),\\
& {\rm supp}_{k\in \mathbb{R}^{3(2i+1)}}\varphi_{2l+1}^{(2i+1)}(k)=S_{2i+1} \mbox{ or }
\emptyset,\quad (l\ge0,\; 0\le i\le l),
\end{align*}
where we set $\varphi_n=\{\varphi_n^{(i)}\}_{i=0}^{\infty}$ by induction for  $n\ge1$, and 
$$ S_i=\{(k_1,\cdots,k_i)\in \mathbb{R}^{3i}| \kappa \le |k_1|\le \Lambda , \cdots , \kappa \le |k_i|\le \Lambda \}. $$ 
Since $\varphi_1=-\hz H_{\rm I}\Omega \in \mathcal{F}^{(1)}$, 
$\varphi_1^{(i)}=0$, 
$(i >1,\; i=0)$. Moreover, since 
\begin{align*}
\varphi_1^{(1)}(k_1)=-\frac{1}{\sqrt{2}(2\pi)^{3/2}\sqrt{\omega(k_1)}E(k_1)}\chi_{[\kappa, \Lambda]}(|k_1|), 
\end{align*}
we have ${\rm supp}_{k_1\in\mathbb{R}^3}\varphi_1^{(1)}(k_1)=S_1$, 
where $E(k)=|k|^2/2m+\omega(k)$. 
Assume that the assumption of the induction holds when $n \le 2l+1,$ $(l\ge 0)$. Then 
\begin{align}
H_0 \varphi_{2l+2}+ (2l+2) H_{\rm I} \varphi_{2l+1} =\sum_{j=1}^{l+1}{2l+2 \choose 2j}E_{2j}\varphi_{2l+2-2j}. \label{H0phi2lplus2}
\end{align}
It is derived that $\varphi_{2l+2}^{(i)}=0$, $(i>2l+2,\; i=0)$ by $(\varphi_n ,\Omega)=\delta_{0 n}$ and (\ref{H0phi2lplus2}).
By the assumption of the induction, $(H_{\rm I}\varphi_{2l+1})^{(2i+1)}=0$, $(0 \le i \le l) $ holds. 
When $1\le q \le l$, 
it holds that 
\begin{align*}
(H_0\varphi_{2l+2})^{(2q)}
&=-2(l+1)(H_{\rm I}\varphi_{2l+1})^{(2q)}+\sum_{j=1}^{l+1}{2l+2 \choose
2j}E_{2j}\varphi_{2l+2-2j}^{(2q)} \\
&
=-\sqrt{2}(l+1)\{\frac{1}{\sqrt{2q}}\sum_{i=1}^{2q}\frac{1}{(2\pi)^{\frac{3}{2}}}\frac{\chi_{[\kappa, \Lambda ]}(|k_i|)}{\sqrt{2}\sqrt{\omega(k_i)}}\varphi_{2l+1}^{(2q-1)}(k_1,\cdots, \hat{k_i},\cdots, k_{2q})\\ 
&  +\sqrt{2q+1}\int dk\frac{1}{(2\pi)^{\frac{3}{2}}}
\frac{\chi_{[\kappa, \Lambda]}(|k|)}{\sqrt{2}\sqrt{\omega(k)}}\varphi_{2l+1}^{(2q+1)}(k, k_1,\cdots, k_{2q})\}+\sum_{j=1}^{l+1}{2l+2 \choose 2j}E_{2j}\varphi_{2l+1-2j}^{(2q)},  
\end{align*} where $\hat{k_i}$ means that $k_i$ is omitted. By the assumption of the induction, the supports of the functions 
\begin{align*}
\frac{1}{(2\pi)^{\frac{3}{2}}}\frac{\chi_{[\kappa, \Lambda ]}(|k_i|)}{\sqrt{2}\sqrt{\omega(k_i)}}\varphi_{2l+1}^{(2q-1)}(k_1,\cdots, \hat{k_i},\cdots, k_{2q}),\quad 
 \int dk\frac{1}{(2\pi)^{\frac{3}{2}}}
 \frac{\chi_{[\kappa, \Lambda]}(|k|)}{\sqrt{2}\sqrt{\omega(k)}}\varphi_{2l+1}^{(2q+1)}(k, k_1,\cdots, k_{2q})
\end{align*}
and $\varphi_{2l+1-2j}^{(2q)}$ are $S_{2q}$ or $\emptyset$. Furthermore, 
\begin{align*}
&(H_0\varphi_{2l+2})^{(2l+2)}(k_1,\cdots ,k_{2l+2})=-2(l+1)(A^+\varphi_{2l+1})^{(2l+2)}(k_1, \cdots ,k_{2l+2}) \\
& = -\sqrt{2(l+1)}\sum_{i=1}^{2l+2}\frac{1}{(2\pi)^{\frac{3}{2}}}\frac{\chi_{[\kappa, \Lambda ]}(|k_i|)}{\sqrt{2}\sqrt{\omega(k_i)}}\varphi_{2l+1}^{(2l+1)}(k_1,\cdots, \hat{k_i},\cdots, k_{2l+2})
\end{align*}  
holds. 
By the assumption of the induction, the support of the right hand side is $S_{2l+2}$ or $\emptyset$. 
Hence we have 
${\rm supp}_{k\in \mathbb{R}^{3\cdot 2i}}\varphi_{2l+2}^{(2i)}(k)=S_{2i}$ 
or 
$\emptyset, (1\le i\le l+1)$.  
We can prove $\varphi_{2l+3}^{(i)}=0, (i > 2l+3,\; i=0)$,     
$\varphi_{2l+3}^{(2i)}=0, (1\le i\le l+1)$, and ${\rm supp}_{k\in \mathbb{R}^{3(2i+1)}}\varphi_{2l+3}^{(2i+1)}(k)=S_{2i+1}$ 
or 
$\emptyset, (0\le i\le l+1)$ in a similar way. 
From the discussion so far, we have 
\begin{align*}
& -2l H_{\rm I} \varphi_{2l-1} +\sum_{j=1}^l{2l \choose 2j}E_{2j}\varphi_{2l-2j}\in {\mathcal F}_{\rm fin}\quad (l \ge 1),\\
&
-(2l+1) H_{\rm I} \varphi_{2l} +\sum_{j=1}^l{2l+1 \choose 2j}E_{2j}\varphi_{2l+1-2j}\in {\mathcal F}_{\rm fin}\quad (l\ge 0).
\end{align*}
Hence we have
\begin{align*}
&\varphi_{2l}=\hz \{-2lH_{\rm I} \varphi _{2l-1}+\sum_{j=1}^{l}{2l \choose 2j} E_{2j} \varphi _{2l-2j}\}+b_{2l}\Omega \quad (l\ge 1), \\
&
\varphi_{2l+1}=\hz \{-(2l+1)H_{\rm I} \varphi _{l}+\sum_{j=1}^{l} {2l+1 \choose 2j}E_{2j} \varphi _{2l+1-2j}\}+b_{2l+1}\Omega \quad  (l\ge 0),
\end{align*}
where $b_{2l}$ and $b_{2l+1}$ are some constants. Since $(\varphi_{2l}, \Omega)=0$, $(l\ge 1)$ and  $(\varphi_{2l+1}, \Omega)=0,$ $(l \ge 0)$, $b_{2l}=b_{2l+1}=0.$ Hence (\ref{phi2l}) and (\ref{phi2lplus1}) are proven . By the discussion so far, (\ref{sumphi2l}) and (\ref{sumphi2l+1}) are also proven . We can derive (\ref{E2l})  by (\ref{phi2l}) and $(\varphi_n ,\Omega)=\delta_{0 n}$.
\end{proof}

\section{Main theorems}
For notational simplicity we set 
$  \hat{\phi}_j=\hat{\phi}(k_j)$ and $\omega_j=\omega(k_j)$ for $k_j \in \mathbb{R}^3, j=1, 2$.  
Let
\begin{align*}
E_j=\frac{|k_j|^2}{2m}+\omega_j, \quad j=1, 2, \quad 
 E_{12}=\frac{|k_1+k_2|^2}{2m}+\omega_1+\omega_2, 
 \quad \omega(r)=\sqrt{r^2+\nu^2}, \quad
 F(r)=\frac{r^2}{2m}+\omega(r). 
\end{align*}

\begin{Theorem}
\label{1}
Let $\kappa>0$. Then $m_{\rm eff}$ is an analytic function of $\alpha^2$ and can be expanded in the following power series for sufficiently small $|\alpha|$ $:$
$$  \frac{m_{\rm eff}}{m}=1+\sum_{n=1}^{\infty}a_n(\Lambda)\alpha ^{2n}. $$  
\end{Theorem}
\begin{proof}
By the power series (\ref{eq:tenkai1}), we have
$$ 
(\psi_g, \psi_g)= (\sum_{n=0}^{\infty}\frac{\alpha ^n}{n!} \varphi_n,\sum_{m=0}^{\infty}\frac{\alpha ^m}{m!} \varphi_m) 
=\sum_{n=0}^{\infty}\sum_{m=0}^{\infty}\frac{\alpha ^{n+m}}{n!m!}(\varphi_n, \varphi_m).   
$$ 
By Lemma \ref{lemma:varphi}, $(\varphi_n, \varphi_m)\neq 0$ if and only if both $n$ and $m$ are even or odd. Then we have
\begin{align*}
(\psi_g, \psi_g)=1+\sum_{n=1}^{\infty}b_n(\Lambda)\alpha ^{2n}.
\end{align*}
From the fact that both $m_{\rm eff}^{-1}$ and $(\psi_g, \psi_g)$ are analytic functions of $\alpha^2$ and  Lemma \ref{formula}, 
we have the following power series: 
\begin{align*}
-\frac{2}{3}\sum_{\mu =1}^3 (P_{{\rm f}\mu}\psi_{g},\psi_{g_{\mu}}\prime(0))=\sum_{n=0}^{\infty}c_n(\Lambda)\alpha^{2n}. 
\end{align*} 
Since $\psi'_{g_{\mu}}(0)$ is an analytic function of $\alpha$, we can write
\begin{align*}
\psi_{g_{\mu}}'(0)= \sum_{n=0}^{\infty} \frac{\alpha ^n}{n!} \Phi_n^{\mu}. \end{align*}
We note that 
\begin{align*}
c_0(\Lambda)&=-\frac{2}{3}\sum_{\mu=1}^3(P_{{\rm f}\mu}\varphi_0, \Phi_0^{\mu})=-\frac{2}{3}\sum_{\mu=1}^3({\rm d}\Gamma(k_{\mu})\varphi_0, \Phi_0^{\mu})=0.
\end{align*}
Hence if $|\alpha|$ is sufficiently small,  then we have the following power series:
\begin{align}
\frac{m_{\rm eff}}{m}&=\frac{(\psi_g, \psi_g)}{(\psi_g, \psi_g)-\frac{2}{3}\sum_{\mu =1}^3 (P_{{\rm f}\mu}\psi_{g}(0), \psi_{g_{\mu}}\prime(0))}  
= \frac{1+\sum_{n=1}^{\infty}b_n(\Lambda)\alpha ^{2n}}{1+\sum_{n=1}^{\infty}(b_n(\Lambda)+c_n(\Lambda))\alpha ^{2n}}  \notag \\
&=(1+\sum_{n=1}^{\infty}b_n(\Lambda)\alpha ^{2n})\sum_{n=0}^{\infty}(-\sum_{l=1}^{\infty}(b_l(\Lambda)+c_l(\Lambda))\alpha^{2l})^n. \label{eq:powerseries}
\end{align}
This proves the theorem.
\end{proof}
\begin{Theorem}\label{2}
There exists strictly positive constant $C$ such that
$ \lim_{\Lambda\to\infty}a_1(\Lambda)=C$. 
\end{Theorem}
\begin{proof}
From (\ref{eq:powerseries}), we have
\begin{align*}
\frac{m_{\rm eff}}{m}=\{1+b_1(\Lambda)\alpha^2+O(\alpha^4)\}[1-\{b_1(\Lambda)+c_1(\Lambda)\}\alpha^2+O(\alpha^4) ]  \
=1-c_1(\Lambda)\alpha^2+O(\alpha^4). 
\end{align*}
Therefore
$a_1(\Lambda)=-c_1(\Lambda)$. 
Since 
\begin{align*}
\sum_{n=1}^{\infty}c_n(\Lambda)\alpha^{2n}=-\frac{2}{3}\sum_{\mu=1}^3(P_{{\rm f}\mu}\psi_g, \psi_{g_{\mu}}'(0)) =-\frac{2}{3}\sum_{\mu=1}^3(P_{{\rm f}\mu} \sum_{n=0}^{\infty}\frac{\alpha^n}{n!}\varphi_n, \sum_{n=0}^{\infty}\frac{\alpha^n}{n!}\Phi_n^{\mu}), 
\end{align*}
we have
\begin{align}
c_1(\Lambda)&=-\frac{2}{3}\sum_{\mu=1}^3\{\frac{1}{0!2!}(P_{{\rm f}\mu}\varphi_0, \Phi_2^{\mu})+\frac{1}{1!1!}(P_{{\rm f}\mu}\varphi_1, \Phi_1^{\mu})+\frac{1}{2!0!}(P_{{\rm f}\mu}\varphi_2, \Phi_0^{\mu})\} \notag \\
&=-\frac{2}{3}\sum_{\mu=1}^3\{(P_{{\rm f}\mu}\varphi_1, \Phi_1^{\mu})+\frac{1}{2}(P_{{\rm f}\mu}\varphi_2, \Phi_0^{\mu})\}.  \label{c_1Lambda}
\end{align}
Substituting $p=0$ into (\ref{pmuderivative1})  and using (\ref{eq:Ederivative}), we have
 \begin{align}
-\frac{1}{m}(P_{{\rm f}\mu}\Psi, \psi_g)+((H_0+\alpha H_{\rm I})\Psi, \psi_{g_{\mu}}'(0))=E(0,\alpha)(\Psi, \psi_{g_{\mu}}'(0)). \label{eq:naiseki4}
\end{align}
In addition, by setting $\alpha=0$, we have
$  -\frac{P_{{\rm f}\mu}}{m}\varphi_0+H_0\Phi_0^{\mu}=0$. 
Since $P_{{\rm f}\mu}\varphi_0={\rm d}\Gamma(k_{\mu})\Omega=0,$ 
$H_0 \Phi_0^{\mu}=0$ holds.
Hence we have
\begin{align}
\Phi_0^{\mu}=c_0\Omega, \quad (c_0 \text{ is some constant.}) \label{eq:phi0mu}
\end{align}
Differentiating both sides of (\ref{eq:naiseki4}) with respect to $\alpha,$ we have
\begin{align*}
&-\frac{1}{m}(P_{{\rm f}\mu}\Psi, s-\frac{d}{d\alpha}\psi_g)+(H(0)\Psi, s-\frac{d}{d\alpha}\psi_{g_{\mu}}'(0))+(H_{\rm I}\Psi, \psi_{g_{\mu}}'(0)) \notag \\
&=E(0,\alpha)(\Psi, s-\frac{d}{d\alpha}\psi_{g_{\mu}}'(0))+\frac{d}{d\alpha}E(0,\alpha)(\Psi, \psi_{g_{\mu}}'(0)). 
\end{align*}
Substituting $\alpha=0$ into both sides, we have
\begin{align*}
(H_0\Psi, \Phi_1^{\mu})=\frac{1}{m}(\Psi, P_{{\rm f}\mu}\varphi_1)-(\Psi, H_{\rm I}\Phi_0^{\mu}).
\end{align*}
Therefore $\Phi_1^{\mu} \in D(H_0)$ and
$$ 
H_0\Phi_1^{\mu}=\frac{1}{m}P_{{\rm f}\mu}\varphi_1-H_{\rm I}\Phi_0^{\mu} =-\frac{1}{m}P_{{\rm f}\mu}\hz H_{\rm I}\Omega-c_0H_{\rm I}\Omega. $$ 
Since $-\frac{1}{m}P_{{\rm f}\mu}\hz H_{\rm I}\Omega-c_0H_{\rm I}\Omega\in {\mathcal F}_{\rm fin}$, we have
\begin{align}
	\Phi_1^{\mu}=-\frac{1}{m}\hz P_{{\rm f}\mu}\hz H_{\rm I}\Omega-c_0\hz H_{\rm I}\Omega+c_1\Omega, \label{Phi1mu}
\end{align}
where $c_1$ is some constant. By  $a_1(\Lambda)=-c_1(\Lambda)$, (\ref{c_1Lambda}),  (\ref{eq:phi0mu}), and (\ref{Phi1mu}),  we have
\begin{align*}
&a_1(\Lambda)=\frac{2}{3}\sum_{\mu=1}^3(P_{{\rm f}\mu}\varphi_1, \Phi_1^{\mu})  
\\ 
&=\frac{2}{3m}\sum_{\mu=1}^3(P_{{\rm f}\mu}\hz H_{\rm I}\Omega, \hz P_{{\rm f}\mu}\hz H_{\rm I}\Omega)
-\frac{2c_0}{3}\sum_{\mu=1}^3(P_{{\rm f}\mu}\hz H_{\rm I}\Omega, \hz H_{\rm I}\Omega)-\frac{2c_1}{3}\sum_{\mu=1}^3(P_{{\rm f}\mu}\hz H_{\rm I}\Omega, \Omega).
\end{align*} 
It is also seen that
$$ \sum_{\mu=1}^3(P_{{\rm f}\mu}\hz H_{\rm I}\Omega, \hz H_{\rm I}\Omega)=\sum_{\mu=1}^3(P_{{\rm f}\mu}\hz H_{\rm I}\Omega, \Omega)=0.$$ 
Thus we have
\begin{align*}
a_1(\Lambda)&=\frac{2}{3m}\sum_{\mu=1}^3(P_{{\rm f}\mu}\hz H_{\rm I}\Omega, \hz P_{{\rm f}\mu}\hz H_{\rm I}\Omega)\notag\\
&=\frac{2}{3m}\sum_{\mu=1}^3(P_{{\rm f}\mu}\hz A^{+}\Omega, \hz P_{ f_{\mu}}\hz A^{+}\Omega) 
=\frac{2}{3m}\int\frac{|\hat{\phi}(k)|^2}{\omega(k)}\frac{|k|^2}{E(k)^3}dk.
\end{align*}
Changing variables into polar coordinate, we have
$$ \frac{2}{3m}\int\frac{|\hat{\phi}(k)|^2}{\omega(k)}\frac{|k|^2}{E(k)^3}dk=\frac{8\pi}{3m(2\pi)^3}\int_{\kappa}^{\Lambda}\frac{r^4}{\omega(r)F(r)^3}dr.$$ 
 Since 
$ \frac{r^4}{\omega(r)F(r)^3}=O(r^{-3})\ (r\to\infty)$, 
the improper integral
$   \int_{\kappa}^{\infty}\frac{r^4}{\omega(r)F(r)^3}dr $ 
converges. It is trivial to see that $\displaystyle \lim_{\Lambda\to\infty}a_1(\Lambda)>0$. Thus the theorem follows.

\end{proof}

\begin{lemma}\label{lemma:ffin}
It follows that
 $\Phi_n^{\mu}\in {\mathcal F}_{\rm fin}$ for $n \in \mathbb{N}\cup\{0\}$.
 \end{lemma}
\begin{proof}
By (\ref{eq:phi0mu}),  we have $\Phi_0^{\mu}\in {\mathcal F}_{\rm fin}.$ Assume that $\Phi_n^{\mu}\in {\mathcal F}_{\rm fin}$
 holds when $n \le k-1.$ Differentiating both sides of (\ref{eq:naiseki4}) $k$ times with respect to $\alpha$ and substituting $\alpha=0,$ we have 
$$ 
-\frac{1}{m}(P_{{\rm f}\mu}\Psi, \varphi_k)+(H_0\Psi, \Phi_k^{\mu})+k(H_{\rm I}\Psi, \Phi_{k-1}^{\mu})=\sum_{j=1}^k {k \choose j}E_j(\Psi, \Phi_{k-j}^{\mu});
$$ 
however, $E_j=0$ when $j$ is odd. Since $\Phi_{k-1}^{\mu} \in {\mathcal F}_{\rm fin}$, $\Phi_{k-1}^{\mu} \in D(H_{\rm I})$ and
$$  (H_0\Psi, \Phi_k^{\mu})=\frac{1}{m}(\Psi, P_{{\rm f}\mu}\varphi_k)-k(\Psi, H_{\rm I}\Phi_{k-1}^{\mu})+(\Psi, \sum_{j=1}^k{ k \choose j}E_j \Phi_{k-j}^{\mu}). $$ 
Thus $\Phi_k^{\mu}\in D(H_0)$ and 
\begin{align}
H_0\Phi_k^{\mu}=\frac{1}{m}P_{{\rm f}\mu}\varphi_k-kH_{\rm I}\Phi_{k-1}^{\mu}+\sum_{j=1}^k{k \choose j}E_j\Phi_{k-j}^{\mu}.
\label{eq:kbibun}
\end{align} 
Since $\varphi_k \in {\mathcal F}_{\rm fin},$  $H_{\rm I}\Phi_{k-1}^{\mu}\in  {\mathcal F}_{\rm fin}$ and $\Phi_{k-j}^{\mu}\in  {\mathcal F}_{\rm fin}$ $(j=1,\cdots, k),$ by the assumption of induction, $H_0\Phi_k^{\mu}\in {\mathcal F}_{\rm fin}.$ Hence $\Phi_n^{\mu}\in {\mathcal F}_{\rm fin}$ holds when $n=k.$
\end{proof}

\begin{lemma}
It holds that 
$\Phi_0^{\mu}=c_0\Omega,$ 
$\Phi_1^{\mu}=-\frac{1}{m}\hz P_{{\rm f}\mu}\hz H_{\rm I}\Omega- c_0\hz H_{\rm I}\Omega+c_1\Omega$ and the recurrence formulas
\begin{align}
& \Phi_{2l}^{\mu}=\hz \{\frac{1}{m}P_{{\rm f}\mu}\varphi_{2l}-2lH_{\rm I}\Phi_{2l-1}^{\mu} 
 +\sum_{j=1}^{l}{2l \choose 2j}E_{2j}\Phi_{2l-2j}^{\mu}\}+c_{2l}\Omega \notag \\
 &\hspace{6cm} (l \ge1, c_{2l}\text{ is some constant.}) \label{phimu2l} \\
 &\Phi_{2l+1}^{\mu}=\hz \{ \frac{1}{m}P_{{\rm f}\mu}\varphi_{2l+1}-(2l+1)H_{\rm I}\Phi_{2l}^{\mu} 
 +\sum_{j=1}^{l}{2l+1 \choose 2j}E_{2j}\Phi_{2l+1-2j}^{\mu}\}+c_{2l+1}\Omega, \notag \\
 &\hspace{6cm}(l \ge 0, c_{2l+1}\text{ is some constant.}) \label{phimu2lplus1}
 \end{align}
\end{lemma}
\begin{proof}
The first and second  expressions are proven  in Theorem \ref{2}. 
From (\ref{eq:kbibun}), 
it follows that 
\begin{align*}
H_0 \Phi_{2l}^{\mu}&=\frac{1}{m}P_{{\rm f}\mu}\varphi_{2l}-2lH_{\rm I}\Phi_{2l-1}^{\mu} 
 +\sum_{j=1}^{l}{2l \choose 2j}E_{2j}\Phi_{2l-2j}^{\mu} 
 \quad(l \ge1) \\
 H_0\Phi_{2l+1}^{\mu}&= \frac{1}{m}P_{{\rm f}\mu}\varphi_{2l+1}-(2l+1)H_{\rm I}\Phi_{2l}^{\mu} 
 +\sum_{j=1}^{l}{2l+1 \choose 2j}E_{2j}\Phi_{2l+1-2j}^{\mu} 
\quad (l \ge 0). 
 \end{align*}
These prove the lemma.
\end{proof}
\begin{lemma} \label{9}
It is proven  that $a_2(\Lambda)$ can be expanded as 
$$ a_2(\Lambda)= \frac{2}{3m}\sum_{j=1}^8{\rm I}_j(\Lambda)+\frac{E_2(\Lambda)}{m}{\rm I}_9(\Lambda)-a_1(\Lambda) {\rm I}_{10}(\Lambda)+a_1(\Lambda)^2, $$ 
where ${\rm I}_j$ are given by~ 
\begin{align*}
\allowdisplaybreaks
&{\rm I}_1(\Lambda)=\frac{1}{4}\int\!\!\!\!\int \!\! dk_1dk_2\frac{|\hat{\phi}_1|^2|\hat{\phi}_2|^2}{\omega_1\omega_2}(\frac{|k_1|^2}{E_1^3}+\frac{|k_2|^2}{E_2^3})(\frac{1}{E_1}+\frac{1}{E_2})\frac{1}{E_{12}}, \\
&
{\rm I}_2(\Lambda)=\frac{1}{8}\int\!\!\!\!\int \!\! dk_1dk_2\frac{|\hat{\phi}_1|^2|\hat{\phi}_2|^2}{\omega_1\omega_2}(\frac{|k_1|^2}{E_1^4}+\frac{|k_2|^2}{E_2^4})\frac{1}{E_{12}}, \\
&
{\rm I}_3(\Lambda)=\frac{1}{8}\int\!\!\!\!\int \!\! dk_1dk_2\frac{|\hat{\phi}_1|^2|\hat{\phi}_2|^2}{\omega_1\omega_2}(\frac{1}{E_1^2}+\frac{1}{E_2^2})(\frac{1}{E_1}+\frac{1}{E_2})\frac{(k_1, k_2)}{E_{12}^2},\\
&
{\rm I}_4(\Lambda)=\frac{1}{4}\int\!\!\!\!\int \!\! dk_1dk_2\frac{|\hat{\phi}_1|^2|\hat{\phi}_2|^2}{\omega_1\omega_2}(\frac{|k_1|^2}{E_1^2}+\frac{|k_2|^2}{E_2^2})(\frac{1}{E_1}+\frac{1}{E_2})\frac{1}{E_{12}^2}, \\
&
{\rm I}_5(\Lambda)=\frac{1}{4}\int\!\!\!\!\int \!\! dk_1dk_2\frac{|\hat{\phi}_1|^2|\hat{\phi}_2|^2}{\omega_1\omega_2E_1^2E_2^2}\frac{(k_1, k_2)}{E_{12}},\\
&
{\rm I}_6(\Lambda)=\frac{1}{8}\int\!\!\!\!\int \!\! dk_1dk_2\frac{|\hat{\phi}_1|^2|\hat{\phi}_2|^2}{\omega_1\omega_2}(\frac{1}{E_1}+\frac{1}{E_2})^2\frac{|k_1|^2+|k_2|^2}{E_{12}^3},\\
&
{\rm I}_7(\Lambda)=\frac{1}{4}\int\!\!\!\!\int \!\! dk_1dk_2\frac{|\hat{\phi}_1|^2|\hat{\phi}_2|^2}{\omega_1\omega_2}(\frac{1}{E_1}+\frac{1}{E_2})^2\frac{(k_1, k_2)}{E_{12}^3},\\
&
{\rm I}_8(\Lambda)=\frac{1}{4}\int\!\!\!\!\int \!\! dk_1dk_2\frac{|\hat{\phi}_1|^2|\hat{\phi}_2|^2}{\omega_1\omega_2}(\frac{1}{E_1}+\frac{1}{E_2})\frac{(k_1, k_2)}{E_{12}^4},\\
&
{\rm I}_9(\Lambda)=\frac{1}{2}\int \frac{|\hat{\phi}(k)|^2|k|^2}{\omega(k)E(k)^4}dk,\\
&
{\rm I}_{10}(\Lambda)=\frac{1}{2}\int \frac{|\hat{\phi}(k)|^2}{\omega(k)E(k)^2}dk. 
\end{align*}
\end{lemma}
The proof of Lemma \ref{9} is given in the next section. 
The asymptotic behaviors of terms ${\rm I}_j(\Lambda) $ as $\Lambda\to\infty$ 
is given in the lemma below. 
Only two terms ${\rm I}_1(\Lambda)$ and ${\rm I}_2(\Lambda)$ logarithmically diverge, and other terms converge as $\Lambda\to\infty$.  
\begin{lemma} \label{oosawa}
(1)-(3) follow the following:
\begin{enumerate}
\item[(1)]
There exist some constants $C_3$ and $C_4$ such that 
$ C_3 \le \lim_{\Lambda\to\infty}\frac{{\rm I}_1(\Lambda)}{\log\Lambda}\le C_4$. 
\item[(2)]
There exist some constants $C_5$ and $C_6$ such that
$C_5 \le \lim_{\Lambda\to\infty}\frac{{\rm I}_2(\Lambda)}{\log\Lambda}\le C_6$.
\item[(3)] For $j=3,4,5,6,7,8$ 
$\lim_{\Lambda\to\infty}|{\rm I}_j(\Lambda)|<\infty$. 
\end{enumerate}
\end{lemma}
The proof of Lemma \ref{oosawa} is technical and also given in the next section.
\begin{lemma}\label{E2logLambda}
It holds that 
\begin{align}
\lim_{\Lambda \to \infty} \frac{E_2(\Lambda)}{\log \Lambda}=-\frac{m}{\pi^2}
\label{eq;e2/loglamabda}.
\end{align}
\end{lemma}

\begin{proof}
From (\ref{E2l}), we have
$E_2(\Lambda)=-\frac{1}{2\pi^2}\int_{\kappa}^{\Lambda}dr\frac{r^2}{\omega(r)F(r)}
$
and 
$  \lim_{r\to\infty}\frac{\frac{r^2}{\omega(r)F(r)}}{\frac{1}{r}}=2m$. 
It implies \eqref{eq;e2/loglamabda}.
\end{proof}
Now we are in the position to state the main theorem in this paper. 
\begin{Theorem}\label{3}
There exist some constants $C_1$ and $C_2$ such that
\begin{align*}
C_1\le \lim_{\Lambda\to\infty}\frac{a_2(\Lambda)}{\log\Lambda}\le C_2. 
\end{align*}
\end{Theorem}
\begin{proof}
We have
$ 
{\rm I}_9(\Lambda)=\frac{1}{4\pi^2}\int_{\kappa}^{\Lambda}\frac{r^4}{\omega(r)F(r)^4}dr$.  
Since 
$  \frac{r^4}{\omega(r)F(r)^4}=O(r^{-5})\  (r\to\infty)$, 
we have 
\begin{align}\label{o2}
\lim_{\Lambda\to\infty}|{\rm I}_9(\Lambda)|<\infty. 
\end{align}
We also have
$ 
{\rm I}_{10}(\Lambda)=\frac{1}{4\pi^2}\int_{\kappa}^{\Lambda}\frac{r^2}{\omega(r)F(r)^2}dr$. 
Then 
$ \frac{r^2}{\omega(r)F(r)^2}=O(r^{-3})\  (r \to\infty)$ 
and 
we also have 
\begin{align}\label{o3}
\lim_{\Lambda\to\infty}|{\rm I}_{10}(\Lambda)|<\infty. 
\end{align}
By \eqref{o2} and \eqref{o3}, Theorem \ref{2} and  Lemmas \ref{9} and 
\ref{oosawa}, and \ref{E2logLambda} we can conclude the theorem.
\end{proof}

\section{Proof of  Lemmas \ref{9} and \ref{oosawa}}
In this section we prove  Lemmas \ref{9} and \ref{oosawa}. 
\subsection{Proof of  Lemma \ref{9}}
From (\ref{eq:powerseries}) and $a_1(\Lambda)=-c_1(\Lambda)$, we have
$
a_2(\Lambda)=-c_2(\Lambda)-b_1(\Lambda)a_1(\Lambda)+a_1(\Lambda)^2$. 
Here 
\begin{align}
b_1(\Lambda)&=(\varphi_1, \varphi_1)=\frac{1}{2}\int \frac{|\hat{\phi(k)|^2}}{\omega(k)E(k)^2}dk,\notag\\
c_2(\Lambda)&=-\frac{2}{3}\{\frac{1}{0!4!}\sum_{\mu=1}^3(P_{{\rm f}\mu}\varphi_0, \Phi_4^{\mu})+\frac{1}{1!3!}\sum_{\mu=1}^3(P_{{\rm f}\mu}\varphi_1, \Phi_3^{\mu}) \notag \\
&+\frac{1}{2!2!}\sum_{\mu=1}^3(P_{{\rm f}\mu}\varphi_2, \Phi_2^{\mu})+\frac{1}{3!1!}\sum_{\mu=1}^3(P_{{\rm f}\mu}\varphi_3, \Phi_1^{\mu})+\frac{1}{4!0!}\sum_{\mu=1}^3(P_{ f_{\mu}}\varphi_4, \Phi_0^{\mu})\} \notag \\
&=-\frac{1}{9}\sum_{\mu=1}^3(P_{{\rm f}\mu}\varphi_1, \Phi_3^{\mu})-\frac{1}{6}\sum_{\mu=1}^3(P_{{\rm f}\mu}\varphi_2, \Phi_2^{\mu})-\frac{1}{9}\sum_{\mu=1}^3(P_{ f_{\mu}}\varphi_3, \Phi_1^{\mu}). \label{eq:c2}
\end{align}
Using recurrence formulas (\ref{phi2l}),  (\ref{phi2lplus1}), (\ref{phimu2l}), and (\ref{phimu2lplus1}),  we have
\begin{align*}
 \varphi_2&=2(\hz H_{\rm I})^2\Omega, \\
 \varphi_3&=-6(\hz H_{\rm I})^3\Omega-3E_2(\hz )^2H_{\rm I}\Omega, \\
\Phi_{2}^{\mu}&=\frac{2}{m}\hz P_{{\rm f}\mu}(\hz H_{\rm I})^2\Omega+\frac{2}{m}\hz H_{\rm I}\hz P_{{\rm f}\mu}\hz H_{\rm I}\Omega \\
&+2c_0(\hz H_{\rm I})^2\Omega-2c_1\hz H_{\rm I}\Omega+c_2\Omega 
\end{align*}
and
\begin{align*}
\Phi_3^{\mu}&=-\frac{6}{m}\hz P_{{\rm f}\mu}(\hz H_{\rm I})^3\Omega-\frac{6}{m}\hz H_{\rm I}\hz P_{{\rm f}\mu}(\hz H_{\rm I})^2\Omega-\frac{6}{m}(\hz H_{\rm I})^2\hz P_{{\rm f}\mu}\hz H_{\rm I}\Omega \\
&-\frac{3}{m}E_2\hz P_{{\rm f}\mu}(\hz )^2H_{\rm I}\Omega-\frac{3}{m}E_2(\hz )^2P_{ f_{\mu}}\hz H_{\rm I}\Omega \\
&-6c_0(\hz H_{\rm I})^3\Omega+6c_1(\hz H_{\rm I})^2\Omega-3c_2\hz H_{\rm I}\Omega-3c_0E_2(\hz )^2H_{\rm I}\Omega+c_3\Omega.\end{align*}
Substituting them into (\ref{eq:c2}), we have
\allowdisplaybreaks
\begin{align*}
&c_2(\Lambda)=-\frac{2}{3m}\{\sum_{\mu=1}^3(P_{{\rm f}\mu}\hz H_{\rm I}\Omega, \hz P_{ f_{\mu}}(\hz H_{\rm I})^3\Omega)+\sum_{\mu=1}^3(P_{{\rm f}\mu}\hz H_{\rm I}\Omega, \hz H_{\rm I}\hz P_{{\rm f}\mu}(\hz H_{\rm I})^2\Omega)  \\
&+\sum_{\mu=1}^3(P_{{\rm f}\mu}\hz H_{\rm I}\Omega, (\hz H_{\rm I})^2\hz P_{ f_{\mu}}\hz H_{\rm I}\Omega)+\sum_{\mu=1}^3(P_{{\rm f}\mu}(\hz H_{\rm I})^2\Omega, \hz P_{ f_{\mu}}(\hz H_{\rm I})^2\Omega)     \\
& +\sum_{\mu=1}^3(P_{ f_{\mu}}(\hz H_{\rm I})^2\Omega, \hz H_{\rm I}\hz P_{ f_{\mu}}\hz H_{\rm I}\Omega)+\sum_{\mu=1}^3(P_{ f_{\mu}}(\hz H_{\rm I})^3\Omega, \hz P_{{\rm f}\mu}\hz H_{\rm I}\Omega) \} \\
& -\frac{E_2}{3m}\{ \sum_{\mu=1}^3(P_{{\rm f}\mu}\hz H_{\rm I}\Omega, \hz P_{{\rm f}\mu}(\hz )^2H_{\rm I}\Omega)+\sum_{\mu=1}^3(P_{{\rm f}\mu}\hz H_{\rm I}\Omega, (\hz )^2P_{{\rm f}\mu}\hz H_{\rm I}\Omega) \\
&+\sum_{\mu=1}^3(P_{{\rm f}\mu}(\hz )^2H_{\rm I}\Omega, \hz P_{{\rm f}\mu}\hz H_{\rm I}\Omega) \} 
-\frac{2c_0}{3}\sum_{\mu=1}^3(P_{{\rm f}\mu}\hz H_{\rm I}\Omega, (\hz H_{\rm I})^3\Omega)   \\
& +\frac{2c_1}{3}\sum_{\mu=1}^3(P_{{\rm f}\mu}\hz H_{\rm I}\Omega, (\hz H_{\rm I})^2\Omega)-\frac{c_2}{3}\sum_{\mu=1}^3(P_{{\rm f}\mu}\hz H_{\rm I}\Omega, \hz H_{\rm I}\Omega) \\
&-\frac{c_0E_2}{3}\sum_{\mu=1}^3(P_{{\rm f}\mu}\hz H_{\rm I}\Omega, (\hz )^2H_{\rm I}\Omega)+\frac{c_3}{9}\sum_{\mu=1}^3(P_{{\rm f}\mu}\hz H_{\rm I}\Omega, \Omega)   \\
&-\frac{2c_0}{3}\sum_{\mu=1}^3(P_{{\rm f}\mu}(\hz H_{\rm I})^2\Omega, (\hz H_{\rm I})^2\Omega)+\frac{2c_1}{3}\sum_{\mu=1}^3(P_{{\rm f}\mu}(\hz H_{\rm I})^2\Omega,  \hz H_{\rm I}\Omega) \\
&-\frac{c_2}{3}\sum_{\mu=1}^3(P_{{\rm f}\mu}(\hz H_{\rm I})^2\Omega, \Omega) -\frac{2c_0}{3}\sum_{\mu=1}^3(P_{{\rm f}\mu}(\hz H_{\rm I})^3\Omega, \hz H_{\rm I}\Omega)     \\
&+\frac{2c_1}{3}\sum_{\mu=1}^3(P_{{\rm f}\mu}(\hz H_{\rm I})^3\Omega, \Omega)-\frac{c_0E_2}{3}\sum_{\mu=1}^3(P_{{\rm f}\mu}(\hz )^2H_{\rm I}\Omega, \hz H_{\rm I}\Omega) \\
&+\frac{c_1E_2}{3}\sum_{\mu=1}^3(P_{{\rm f}\mu}(\hz )^2H_{\rm I}\Omega, \Omega)=\sum_{j=1}^{21}(j).  
\end{align*}
We estimate 21 terms $(1)-(21)$ above. 
We can however directly see that $0=(10)=\cdots=(21)$: 
\begin{align*}
&\sum_{\mu=1}^3(P_{{\rm f}\mu}\hz H_{\rm I}\Omega, (\hz H_{\rm I})^3\Omega)
=\sum_{\mu=1}^3(P_{{\rm f}\mu}\hz H_{\rm I}\Omega, (\hz H_{\rm I})^2\Omega)
=\sum_{\mu=1}^3(P_{{\rm f}\mu}\hz H_{\rm I}\Omega, \hz H_{\rm I}\Omega)\\
&
=\sum_{\mu=1}^3(P_{{\rm f}\mu}\hz H_{\rm I}\Omega, (\hz )^2H_{\rm I}\Omega)
=\sum_{\mu=1}^3(P_{{\rm f}\mu}\hz H_{\rm I}\Omega, \Omega) 
=\sum_{\mu=1}^3(P_{{\rm f}\mu}(\hz H_{\rm I})^2\Omega, (\hz H_{\rm I})^2\Omega)\\
&
=\sum_{\mu=1}^3(P_{{\rm f}\mu}(\hz H_{\rm I})^2\Omega,  \hz H_{\rm I}\Omega)
=\sum_{\mu=1}^3(P_{{\rm f}\mu}(\hz H_{\rm I})^2\Omega, \Omega)
=\sum_{\mu=1}^3(P_{{\rm f}\mu}(\hz H_{\rm I})^3\Omega, \hz H_{\rm I}\Omega)\\
&
=\sum_{\mu=1}^3(P_{{\rm f}\mu}(\hz H_{\rm I})^3\Omega, \Omega)
=\sum_{\mu=1}^3(P_{{\rm f}\mu}(\hz )^2H_{\rm I}\Omega, \hz H_{\rm I}\Omega) 
=\sum_{\mu=1}^3(P_{{\rm f}\mu}(\hz )^2H_{\rm I}\Omega, \Omega)=0.
\end{align*}
We can compute remaining terms $(1)-(9)$ as 
\begin{align*}
(1)\  & \sum_{\mu=1}^3(P_{{\rm f}\mu}\hz H_{\rm I}\Omega, \hz P_{{\rm f}\mu}(\hz H_{\rm I})^3\Omega) =\sum_{\mu=1}^3(A^{+}\hz P_{{\rm f}\mu}\hz P_{{\rm f}\mu}\hz A^{+}\Omega, (\hz A^{+})^2\Omega)\\ 
&= \frac{1}{8}\int\!\!\!\!\int\!\!dk_1dk_2\frac{|\hat{\phi}_1|^2|\hat{\phi}_2|^2}{\omega_1\omega_2}(\frac{|k_1|^2}{E_1^3}+\frac{|k_2|^2}{E_2^3})(\frac{1}{E_1}+\frac{1}{E_2})\frac{1}{E_{12}},\\
(2)\  
&\sum_{\mu=1}^3(P_{{\rm f}\mu}\hz H_{\rm I}\Omega, \hz H_{\rm I}\hz P_{{\rm f}\mu}(\hz H_{\rm I})^2\Omega) =\sum_{\mu=1}^3(A^{+}\hz P_{{\rm f}\mu}\hz A^{+}\Omega,  \hz P_{{\rm f}\mu}(\hz A^{+})^2\Omega) \\
&=\frac{1}{8}\int\!\!\!\!\int\!\!dk_1dk_2\frac{|\hat{\phi}_1|^2|\hat{\phi}_2|^2}{\omega_1\omega_2}(\frac{|k_1|^2}{E_1^2}+\frac{|k_2|^2}{E_2^2})(\frac{1}{E_1}+\frac{1}{E_2})\frac{1}{E_{12}^2} \\
&\hspace{1cm}+\frac{1}{8}\int\!\!\!\!\int\!\!dk_1dk_2\frac{|\hat{\phi}_1|^2|\hat{\phi}_2|^2}{\omega_1\omega_2}(\frac{1}{E_1^2}
+\frac{1}{E_2^2})(\frac{1}{E_1}+\frac{1}{E_2})\frac{(k_1, k_2)}{E_{12}^2}. \\
(3)\  &\sum_{\mu=1}^3(P_{{\rm f}\mu}\hz H_{\rm I}\Omega, (\hz H_{\rm I})^2\hz P_{{\rm f}\mu}\hz H_{\rm I}\Omega) =\sum_{\mu=1}^3(A^{+}\hz P_{{\rm f}\mu}\hz A^{+}\Omega,  \hz A^{+}\hz P_{{\rm f}\mu}\hz A^{+}\Omega)  \\
&=\frac{1}{8}\int\!\!\!\!\int\!\!dk_1dk_2\frac{|\hat{\phi}_1|^2|\hat{\phi}_2|^2}{\omega_1\omega_2}(\frac{|k_1|^2}{E_1^4}+\frac{|k_2|^2}{E_2^4})\frac{1}{E_{12}} 
+\frac{1}{4}\int\!\!\!\!\int\!\!dk_1dk_2\frac{|\hat{\phi}_1|^2|\hat{\phi}_2|^2}{\omega_1\omega_2E_1^2E_2^2}\frac{(k_1, k_2)}{E_{12}}.  \\
(4)\  &\sum_{\mu=1}^3(P_{{\rm f}\mu}(\hz H_{\rm I})^2\Omega, \hz P_{{\rm f}\mu}(\hz H_{\rm I})^2\Omega) 
=\sum_{\mu=1}^3(P_{{\rm f}\mu}(\hz A^{+})^2\Omega, \hz P_{{\rm f}\mu}(\hz A^{+})^2\Omega)\\
& =\frac{1}{8}\int\!\!\!\!\int\!\!dk_1dk_2\frac{|\hat{\phi}_1|^2|\hat{\phi}_2|^2}{\omega_1\omega_2}(\frac{1}{E_1}+\frac{1}{E_2})^2\frac{|k_1|^2+|k_2|^2}{E_{12}^3}
+\frac{1}{4}\int\!\!\!\!\int\!\!dk_1dk_2\frac{|\hat{\phi}_1|^2|\hat{\phi}_2|^2}{\omega_1\omega_2}(\frac{1}{E_1}+\frac{1}{E_2})^2\frac{(k_1, k_2)}{E_{12}^3}. \\
(5)\  
&\sum_{\mu=1}^3(P_{{\rm f}\mu}(\hz H_{\rm I})^2\Omega, \hz H_{\rm I}\hz P_{{\rm f}\mu}\hz H_{\rm I}\Omega) 
=\sum_{\mu=1}^3(P_{{\rm f}\mu}(\hz A^{+})^2\Omega, \hz A^{+}\hz P_{{\rm f}\mu}\hz A^{+}\Omega) \\
&=\frac{1}{8}\int\!\!\!\!\int\!\!dk_1dk_2\frac{|\hat{\phi}_1|^2|\hat{\phi}_2|^2}{\omega_1\omega_2}(\frac{|k_1|^2}{E_1^2}+\frac{|k_2|^2}{E_2^2})(\frac{1}{E_1}+\frac{1}{E_2})\frac{1}{E_{12}^2} +\frac{1}{4}\int\!\!\!\!\int\!\!dk_1dk_2\frac{|\hat{\phi}_1|^2|\hat{\phi}_2|^2}{\omega_1\omega_2}(\frac{1}{E_1}+\frac{1}{E_2})\frac{(k_1, k_2)}{E_{12}^4}.\\
(6)\ 
&\sum_{\mu=1}^3(P_{{\rm f}\mu}(\hz H_{\rm I})^3\Omega, \hz P_{{\rm f}\mu}\hz H_{\rm I}\Omega) 
=\sum_{\mu=1}^3((\hz A^{+})^2\Omega, A^{+}\hz P_{{\rm f}\mu}\hz P_{{\rm f}\mu}\hz A^{+}\Omega) \\
&=\frac{1}{8}\int\!\!\!\!\int\!\!dk_1dk_2\frac{|\hat{\phi}_1|^2|\hat{\phi}_2|^2}{\omega_1\omega_2}(\frac{|k_1|^2}{E_1^3}+\frac{|k_2|^2}{E_2^3})(\frac{1}{E_1}+\frac{1}{E_2})\frac{1}{E_{12}}.\\
(7)\  
&\sum_{\mu=1}^3(P_{{\rm f}\mu}\hz H_{\rm I}\Omega, \hz P_{{\rm f}\mu}(\hz )^2H_{\rm I}\Omega)=\frac{1}{2}\int\frac{|\hat{\phi}(k)|^2|k|^2}{\omega(k)E(k)^4}dk.\\
(8)\ 
&\sum_{\mu=1}^3(P_{{\rm f}\mu}\hz H_{\rm I}\Omega, (\hz )^2P_{{\rm f}\mu}\hz H_{\rm I}\Omega) =\frac{1}{2}\int\frac{|\hat{\phi}(k)|^2|k|^2}{\omega(k)E(k)^4}dk.\\
(9)\ 
&\sum_{\mu=1}^3(P_{{\rm f}\mu}(\hz )^2H_{\rm I}\Omega, \hz P_{{\rm f}\mu}\hz H_{\rm I}\Omega)
=\frac{1}{2}\int\frac{|\hat{\phi}(k)|^2|k|^2}{\omega(k)E(k)^4}dk.  
\end{align*}
Thus the lemma follows.
\qed
\subsection{Proof of Lemma \ref{oosawa}}
\subsubsection{Proof of  
$C_3\leq \lim_{\Lambda\to\infty}\frac{{\rm I}_1(\Lambda)}{\log\Lambda}$ and
$C_5\leq \lim_{\Lambda\to\infty}\frac{{\rm I}_2(\Lambda)}{\log\Lambda}$}
Changing variables to polar coordinates, we have
\begin{align*}
{\rm I}_1(\Lambda)=\frac{2\pi^2}{(2\pi)^6}\int_{-1}^{1}dz\int_{\kappa}^{\Lambda}dr_1\int_{\kappa}^{\Lambda}\frac{r_1^2r_2^2}{\omega(r_1)\omega(r_2)}(\frac{r_1^2}{F(r_1)^3}+\frac{r_2^2}{F(r_2)^3}) 
(\frac{1}{F(r_1)}+\frac{1}{F(r_2)})\frac{1}{L(r_1, r_2, z)}dr_2,
\end{align*}
where 
\begin{align*}
L(r_1, r_2, z)=\frac{r_1^2+r_2^2+2r_1r_2z}{2m}+\omega(r_1)+\omega(r_2).
\end{align*}
We define $h(r_1, r_2),$ $h_1(r_1, r_2),$ $h_2(r_1, r_2),$ $S(\Lambda),$ $S_1(\Lambda)$ and $S_2(\Lambda)$ as
\begin{align*}
&h(r_1, r_2)=\frac{r_1^2r_2^2}{\omega(r_1)\omega(r_2)}(\frac{r_1^2}{F(r_1)^3}+\frac{r_2^2}{F(r_2)^3})(\frac{1}{F(r_1)}+\frac{1}{F(r_2)})\frac{1}{L(r_1, r_2, 1)},\\
&h_1(r_1, r_2)=\frac{r_1^2r_2^4}{\omega(r_1)\omega(r_2)F(r_2)^4L(r_1, r_2, 1)},\\
&h_2(r_1, r_2)=h(r_1, r_2)-h_1(r_1, r_2),\\
&S(\Lambda)=\int_{\kappa}^{\Lambda}\int_{\kappa}^{\Lambda}h(r_1, r_2)dr_1dr_2,\\
&S_1(\Lambda)=\int_{\kappa}^{\kappa+1}dr_2\int_{\kappa+1+\nu+m}^{\Lambda}h_1(r_1, r_2)dr_1,\\
&S_2(\Lambda)=\int_{\kappa}^{\kappa+1}dr_2\int_{\kappa}^{\kappa+1+\nu+m}h_1(r_1,  r_2)dr_1+\int_{\kappa+1}^{\Lambda}dr_1\int_{\kappa}^{\Lambda}h_1(r_1,
 r_2)dr_2
+\int_{\kappa}^{\Lambda}\int_{\kappa}^{\Lambda}h_2(r_1, r_2)dr_1dr_2.
\end{align*}
Then
\begin{align}
\frac{4\pi^2}{(2\pi)^6}S(\Lambda)\le{\rm I}_1(\Lambda). \label{eq:I1inferior1}
\end{align}
In addition, 
$S(\Lambda)=S_1(\Lambda)+S_2(\Lambda)$ 
 follows. Since $h_1(r_1,r_2)>0$ and $h_2(r_1, r_2)>0$, $S_2(\Lambda)>0$. Hence
\begin{align}
S(\Lambda)>S_1(\Lambda). \label{eq:SLambdaenequality1}
\end{align}
Let $r_2$ satisfy $\kappa\le r_2\le\kappa+1.$ Suppose that $\kappa+1+\nu+m \le r_1\le\Lambda$. Since $\nu<r_1$, $r_1^2+\nu^2<2r_1^2$ holds.  
Therefore we have
$ \omega(r_1)<\sqrt{2}r_1$. 
Since $r_2<r_1$, we  have 
$r_1r_2<r_1^2$ and $r_2^2<r_1^2$.
Thus
$
L(r_1, r_2, 1)<2(\frac{1}{m}+\sqrt{2})r_1^2$. 
So, 
\begin{align*}
\int_{\kappa+1+\nu+m}^{\Lambda}\frac{r_1^2}{\omega(r_1)L(r_1, r_2, 1)}dr_1>\frac{1}{2\sqrt{2}(\frac{1}{m}+\sqrt{2})}\int_{\kappa+1+\nu+m}^{\Lambda}
\!\!\! 
\!\!\! \!\!\! 
\!\!\! \!\!\! r_1^{-1}dr_1 
=\frac{1}{2\sqrt{2}(\frac{1}{m}+2)}(\log\Lambda-\log(\kappa+1+\nu+m))
\end{align*}
follows. When $\displaystyle \kappa\le r_2 \le \kappa+1,$ we have
\begin{align*}
&\omega(r_2)\le\sqrt{(\kappa+1)^2+\nu^2},\quad 
F(r_2)\le\frac{(\kappa+1)^2}{2m}+\sqrt{(\kappa+1)^2+\nu^2}.
\end{align*} 
Then 
\begin{align}
S_1(\Lambda)&=\int_{\kappa}^{\kappa+1}dr_2\frac{r_2^4}{\omega(r_2)F(r_2)^4}\int_{\kappa+1+\nu+m}^{\Lambda}\frac{r_1^2}{\omega(r_1)L(r_1, r_2, 1)}dr_1 \notag\\
&>\frac{K}{2\sqrt{2}(\frac{1}{m}+\sqrt{2})}(\log\Lambda-\log(\kappa+1+\nu+m))\int_{\kappa}^{\kappa+1}r_2^4dr_2 \notag \\
&=\frac{K\{(\kappa+1)^5-\kappa^5\}}{10\sqrt{2}(\frac{1}{m}+\sqrt{2})}(\log\Lambda-\log(\kappa+1+\nu+m)), \label{eq:SlargerthanS11}
\end{align}
where 
$$K=\frac{1}{\sqrt{(\kappa+1)^2+\nu^2}(\frac{(\kappa+1)^2}{2m}+\sqrt{(\kappa+1)^2+\nu^2})^4}.$$
From (\ref{eq:I1inferior1}), (\ref{eq:SLambdaenequality1}), and (\ref{eq:SlargerthanS11}), $C_3\leq \lim_{\Lambda\to\infty}\frac{{\rm I}_1(\Lambda)}{\log\Lambda}$
follows.

The proof of 
 $C_5\leq \lim_{\Lambda\to\infty}\frac{{\rm I}_2(\Lambda)}{\log\Lambda}$
  is similar to that of $C_3\leq \lim_{\Lambda\to\infty}\frac{{\rm I}_1(\Lambda)}{\log\Lambda}$. Then we omit it. 
\qed

\subsubsection{Proof of $\lim_{\Lambda\to\infty}\frac{{\rm I}_1(\Lambda)}{\log\Lambda}\leq C_4$}
We redefine $h(r_1, r_2),$ $h_1(r_1, r_2),$ $h_2(r_1, r_2),$ $S_1(\Lambda),$ $S_2(\Lambda),$ and $S(\Lambda)$ as 
\begin{align*}
&h(r_1, r_2)=\frac{r_1^2r_2^2}{\omega(r_1)\omega(r_2)}(\frac{r_1^2}{F(r_1)^3}+\frac{r_2^2}{F(r_2)^3})(\frac{1}{F(r_1)}+\frac{1}{F(r_2)})\frac{1}{L(r_1, r_2, -1)},\\
&h_1(r_1, r_2)=\frac{r_1^2r_2^4}{\omega(r_1)\omega(r_2)F(r_2)^4L(r_1, r_2, -1)},\\
&h_2(r_1, r_2)=\frac{r_1^2r_2^4}{\omega(r_1)\omega(r_2)F(r_1)F(r_2)^3L(r_1, r_2, -1)},\\
& S_1(\Lambda)=\int_{\kappa}^{\Lambda}\int_{\kappa}^{\Lambda}h_1(r_1, r_2)dr_1dr_2,\\
&S_2(\Lambda)=\int_{\kappa}^{\Lambda}\int_{\kappa}^{\Lambda}h_2(r_1, r_2)dr_1dr_2,\\
&S(\Lambda)=\int_{\kappa}^{\Lambda}\int_{\kappa}^{\Lambda}h(r_1, r_2)dr_1dr_2.
\end{align*}
Then we have
\begin{align}
{\rm I}_1(\Lambda)\le\frac{4\pi^2}{(2\pi)^6}S(\Lambda).\label {eq:I1superior}
\end{align}
Since 
$h(r_1, r_2)=h_1(r_1, r_2)+h_1(r_2, r_1)+h_2(r_1, r_2)+h_2(r_2, r_1)$,  we have
\begin{align}
S(\Lambda)=2(S_1(\Lambda)+S_2(\Lambda)). \label{eq;sums}
\end{align} 
Let $B$ be
$B=\frac{(\kappa+1)^3-\kappa^3}{6\nu^2}$. 
Since $\nu<\omega(r_1)$ and $2\nu< L(r_1, r_2, -1),$
\begin{align}
\int_{\kappa}^{\kappa+1}\frac{r_1^2}{\omega(r_1)L(r_1, r_2, -1)}dr_1<B \label{eq: lessthanB1}
\end{align}
follows. Let $Y$ be $Y=2r_2+\kappa+1$. Since $r_1<\omega(r_1)$ and $r_1<L(r_1, r_2, -1),$
\begin{align}
\int_{\kappa+1}^Y \frac{r_1^2}{\omega(r_1)L(r_1, r_2, -1)}dr_1<2r_2 \label{eq:lessthanr2}
\end{align}
holds. When $Y\le r_1,$ since $2r_2<r_1$, we have  $r_1-r_2>r_1/2$. Then
$L(r_1, r_2, -1)>\frac{(r_1-r_2)^2}{2m}>\frac{r_1^2}{8m}$. 
So,
\begin{align}
\int_Y^{\Lambda}\frac{r_1^2}{\omega(r_1)L(r_1, r_2,-1)}dr_1 &<8m\int_Y^{\Lambda}\frac{dr_1}{r_1}    <8m\log\Lambda \label{eq:lessthan8mloglambda}
\end{align}
follows. From (\ref{eq: lessthanB1}), (\ref{eq:lessthanr2}), and
 (\ref{eq:lessthan8mloglambda}), we have 
$$\int_{\kappa}^{\Lambda}\frac{r_1^2}{\omega(r_1)L(r_1, r_2, -1)}dr_1<B+2r_2+8m\log\Lambda.$$
Using this, we see that
$
S_1(\Lambda) <\int_{\kappa}^{\Lambda}(B+2r_2+8m\log\Lambda)\frac{r_2^4}{\omega(r_2)F(r_2)^4}dr_2
$. 
Since $r_2<\omega(r_2)$ and $r_2^2/2m<F(r_2)$, we have
\begin{align}
S_1(\Lambda)&<\int_{\kappa}^{\Lambda}(2r_2+B+8m\log\Lambda)\frac{16m^4}{r_2^5}dr_2 =\frac{32m^4}{3}(\frac{1}{\kappa^3}-\frac{1}{\Lambda^3})+4m^4(B+8m\log\Lambda)(\frac{1}{\kappa^4}-\frac{1}{\Lambda^4})  \notag \\
&<\frac{32m^5}{\kappa^4}\log\Lambda+\frac{32m^4}{3\kappa^3}+\frac{4m^4B}{\kappa^4}. \label{eq:I1S1sup}
\end{align}
Since $r_1<L(r_1, r_2, -1)$, 
\begin{align*}
 \int_{\kappa}^{\Lambda}\frac{r_1^2}{\omega(r_1)F(r_1)L(r_1, r_2, -1)}dr_1&<2m\int_{\kappa}^{\Lambda}\frac{dr_1}{r_1^2}
<\frac{2m}{\kappa}.
\end{align*}
Then we have 
\begin{align}
S_2(\Lambda)&=\int_{\kappa}^{\Lambda}dr_2 \frac{r_2^4}{\omega(r_2)F(r_2)^3}\int_{\kappa}^{\Lambda}\frac{r_1^2}{\omega(r_1)F(r_1)L(r_1, r_2, -1)}dr_1
<\frac{2m}{\kappa}\int_{\kappa}^{\Lambda} \frac{r_2^4}{\omega(r_2)F(r_2)^3}dr_2 \notag\\
&<\frac{16m^4}{\kappa}\int_{\kappa}^{\Lambda}\frac{dr_2}{r_2^3} 
<\frac{8m^4}{\kappa^3}.   \label{eq:I1S2sup}
\end{align}
From (\ref{eq;sums}), (\ref{eq:I1S1sup}), and (\ref{eq:I1S2sup}), it follows that 
\begin{align}
S(\Lambda)<\frac{64m^5}{\kappa^4}\log\Lambda+\frac{112m^4}{3\kappa^3}+\frac{8m^2B}{\kappa^4} \label{eq:I1supS}. 
\end{align}
From (\ref{eq:I1superior}) and (\ref{eq:I1supS}), the lemma follows.
\qed

\subsubsection{Proof of  
$\lim_{\Lambda\to\infty}\frac{{\rm I}_2(\Lambda)}{\log\Lambda}\leq C_6$}
We redefine $h(r_1, r_2)$, $h_1(r_1, r_2)$, $S(\Lambda),$ and $S_1(\Lambda)$ as
\begin{align*}
&h(r_1, r_2)=\frac{r_1^2r_2^2}{\omega({r_1})\omega({r_2})}(\frac{r_1^2}{F(r_1)^4}+\frac{r_2^2}{F(r_2)^4})\frac{1}{L(r_1, r_2, -1)},\\
&
h_1(r_1, r_2)=\frac{r_1^2r_2^4}{\omega(r_1)\omega(r_2)F(r_2)^4L(r_1, r_2, -1)},\\
&
S(\Lambda)=\int_{\kappa}^{\Lambda}\int_{\kappa}^{\Lambda}h(r_1, r_2)dr_1dr_2,\\
&S_1(\Lambda)=\int_{\kappa}^{\Lambda}\int_{\kappa}^{\Lambda}h_1(r_1, r_2)dr_1dr_2.
\end{align*}
We have
\begin{align*}
{\rm I}_2(\Lambda)=\frac{\pi^2}{(2\pi)^6}\int_{-1}^1dz\int_{\kappa}^{\Lambda}dr_1\int_{\kappa}^{\Lambda}\frac{r_1^2r_2^2}{\omega({r_1})\omega({r_2})} 
(\frac{r_1^2}{F(r_1)^4}+\frac{r_2^2}{F(r_2)^4})\frac{1}{L(r_1, r_2, z)}dr_2.
\end{align*}
Then
\begin{align} 
{\rm I}_2(\Lambda)\le\frac{2\pi^2}{(2\pi)^6}S(\Lambda) \label{eq:I4superior}
\end{align}
holds. Since $h(r_1, r_2)=h_1(r_1, r_2)+h_1(r_2, r_1)$, we have 
\begin{align}
S(\Lambda)=2S_1(\Lambda). \label{eq:Sequal2S1}
\end{align}
We have
$$\int_{\kappa}^{\Lambda}\frac{r_1^2}{\omega(r_1)L(r_1, r_2, -1)}dr_1<B+2r_2+8m\log\Lambda$$
in the same way as the proof of 
$ \lim_{\Lambda\to\infty}\frac{{\rm I}_1(\Lambda)}{\log\Lambda}\leq C_4$.
 Since $r_2<\omega(r_2)$  and $r_2^2/2m<F(r_2),$ we have
\begin{align}
S_1(\Lambda)&<\int_{\kappa}^{\Lambda}(B+2r_2+8m\log\Lambda)\frac{r_2^4}{\omega(r_2)F(r_2)^4}dr_2  \notag \\
&<32m^4\int_{\kappa}^{\Lambda}\frac{dr_2}{r_2^4}+16m^4(B+8m\log\Lambda)\int_{\kappa}^{\Lambda}\frac{dr_2}{r_2^5} 
<\frac{32m^5}{\kappa^4}\log\Lambda+\frac{32m^4}{3\kappa^3}+\frac{4m^4B}{\kappa^4}. \label{eq:S1islessthanconstant}
\end{align}
From (\ref{eq:I4superior}), (\ref{eq:Sequal2S1}), and (\ref{eq:S1islessthanconstant}), the lemma follows.
\qed

\subsubsection{Proof of  $\lim_{\Lambda\to\infty}\frac{{\rm I}_3(\Lambda)}{\log\Lambda}=0$}
We define $h(r_1, r_2, z)$ as
\begin{align*}
h(r_1, r_2, z)=\frac{zr_1^3r_2^3}{\omega(r_1)\omega(r_2)}(\frac{1}{F(r_1)^2}+\frac{1}{F(r_2)^2})(\frac{1}{F(r_1)}+\frac{1}{F(r_2)})\frac{1}{L(r_1,r_2, z)^2}  
\end{align*}
and redefine $S(\Lambda)$ as
\begin{align*}
S(\Lambda)=
\int_{-1}^1dz\int_{\kappa}^{\Lambda}dr_2\int_{\kappa}^{\Lambda}h(r_1, r_2, z)dr_1. 
\end{align*}
Then we have
$
 {\rm I}_3(\Lambda)=\frac{\pi^2}{(2\pi)^6}S(\Lambda)$.  
We divide $S(\Lambda)$ in the following way.
\begin{align}
S(\Lambda)&=\int_{-1}^0 dz\int_{\kappa}^{\Lambda}dr_2\int_{\kappa}^{\Lambda}h(r_1, r_2, z)dr_1 +\int_0^1dz\int_{\kappa}^{\Lambda}dr_2\int_{\kappa}^{\Lambda}h(r_1, r_2, z)dr_1 
\notag \\
&=-\int_{1}^0dz\int_{\kappa}^{\Lambda}dr_2\int_{\kappa}^{\Lambda}h(r_1, r_2, -z)dr_1+\int_0^1dz\int_{\kappa}^{\Lambda}dr_2\int_{\kappa}^{\Lambda}h(r_1, r_2, z)dr_1 \notag \\
&=\int_{0}^1dz\int_{\kappa}^{\Lambda}dr_2\int_{\kappa}^{\Lambda}(h(r_1, r_2, z)+h(r_1, r_2, -z))dr_1 \notag \\
&=\int_{0}^1dz\int_{\kappa}^{\Lambda}dr_2\int_{\kappa}^{\Lambda}g(r_1, r_2, z)dr_1, \label{SLambdadivide}
\end{align}    
where 
\begin{align}
&g(r_1, r_2, z)=h(r_1, r_2, z)+h(r_1, r_2, -z) \notag\\
&=-\frac{2z^2r_1^4r_2^4}{m\omega(r_1)\omega(r_2)}(\frac{1}{F(r_1)^2}+\frac{1}{F(r_2)^2})(\frac{1}{F(r_1)}+\frac{1}{F(r_2)})
\frac{
(\frac{r_1^2+r_2^2}{m}+2\omega(r_1)+2\omega(r_2))}
{L(r_1, r_2, z)^2L(r_1, r_2, -z)^2}. \label{I3g}
\end{align}
Since $g(r_1, r_2, z)\le0$, $S(\Lambda)$ is decreasing in $\Lambda$.
\begin{align}
S(\Lambda)&=\int_0^1dz\int_{\kappa}^{\Lambda}dr_2\int_{r_2}^{\Lambda}g(r_1, r_2, z)dr_1+\int_0^1dz\int_{\kappa}^{\Lambda}dr_1\int_{r_1}^{\Lambda}g(r_1, r_2, z)dr_2
 \notag \\
&=2\int_0^1dz\int_{\kappa}^{\Lambda}dr_2\int_{r_2}^{\Lambda}g(r_1, r_2, z)dr_1 \notag \\
&=2\int_0^1dz\int_{\kappa}^{\Lambda}dr_2\int_{r_2}^{2r_2}g(r_1, r_2, z)dr_1+2\int_0^1dz\int_{\kappa}^{\Lambda}dr_2\int_{2r_2}^{\Lambda}g(r_1, r_2, z)dr_1. \label{I3SLambdaintegral}
\end{align} 
Since $\kappa\le r_1$, we have  $1 \le r_1/\kappa$.
Hence   
$ r_1^2+\nu^2\le \frac{\kappa^2+\nu^2}{\kappa^2}r_1^2$. 
Therefore we have
$   \omega(r_1)\le \frac{\sqrt{\kappa^2+\nu^2}}{\kappa}r_1$, and similarly
$ \omega(r_2)\le \frac{\sqrt{\kappa^2+\nu^2}}{\kappa}r_2$. 
When $0\le z \le 1, $ we have
\begin{align}
L(r_1, r_2, z)>\frac{r_1^2}{2m}. \label{I3Lzlarger} 
\end{align}
Then
\begin{align}
 L(r_1, r_2, -z)
&=\frac{(r_1-r_2)^2+2r_1r_2(1-z)}{2m}+\omega(r_1)+\omega(r_2) 
>\omega(r_1)  
>r_1. \label{I3Lminuszlarger}
\end{align}
When $r_2\le r_1$, 
we have 
$ \frac{r_2^2}{m}\le\frac{r_1^2}{m}$. Then 
it holds that
\begin{align*}
2\omega(r_1)\le\frac{2r_1^2}{\kappa^2}\sqrt{\kappa^2+\nu^2},
\quad 
 2\omega(r_2)\le\frac{2r_2}{\kappa}\sqrt{\kappa^2+\nu^2}\le\frac{2r_1^2}{\kappa^2}\sqrt{\kappa^2+\nu^2}.
 \end{align*}
Thus we have
\begin{align}
\frac{r_1^2}{m}+\frac{r_2^2}{m}+2\omega(r_1)+2\omega(r_2) \le 2(\frac{1}{m}+\frac{2\sqrt{\kappa^2+\nu^2}}{\kappa^2})r_1^2. \label{I3inequality3}
\end{align}
When $r_2 \le r_1 \le 2r_2$, we have
\begin{align}
\frac{1}{F(r_1)}+\frac{1}{F(r_2)}<\frac{10m}{r_1^2}, \quad
\frac{1}{F(r_1)^2}+\frac{1}{F(r_2)^2}<\frac{68m^2}{r_1^4}. \label{I3inequality4}  
\end{align}
From (\ref{I3g}), (\ref{I3Lzlarger}), (\ref{I3Lminuszlarger}), (\ref{I3inequality3}), and
(\ref{I3inequality4}), it follows that
\begin{align*}
-g(r_1, r_2, z)&\le\frac{2}{m}\frac{z^2r_1^4r_2^4}{r_1r_2}\frac{68m^2}{r_1^4}\frac{10m}{r_1^2}2(\frac{1}{m}+\frac{2\sqrt{\kappa^2+\nu^2}}{\kappa^2}
)r_1^2(\frac{2m}{r_1^2})^2\frac{1}{r_1^2}\\
&=68\cdot10\cdot16m^4(\frac{1}{m}+\frac{2\sqrt{\kappa^2+\nu^2}}{\kappa^2})\frac{z^2r_2^3}{r_1^7}.
\end{align*}
Hence 
\begin{align}
-2\int_{0}^1dz\int_{\kappa}^{\Lambda}dr_2\int_{r_2}^{2r_2}g(r_1, r_2, z)dr_1 \le1190\frac{m^4}{\kappa^2}(\frac{1}{m}+\frac{2\sqrt{\kappa^2+\nu^2}}{\kappa^2}).\label{I3integral1}
\end{align}
 When $2r_2\le r_1$, we have
\begin{align}
\frac{1}{F(r_1)}+\frac{1}{F(r_2)}<\frac{5m}{2r_2^2},\quad 
\frac{1}{F(r_1)^2}+\frac{1}{F(r_2)^2}<\frac{17m^2}{4r_2^4}. \label{I3inequality7}
\end{align}
Since $ r_2 \le r_1/2$, we have $r_1/2\le r_1-r_2$. Then we have
\begin{align}
L(r_1, r_2, \pm z)=\frac{(r_1-r_2)^2+2r_1r_2(1\pm z)}{2m}+\omega(r_1)+\omega(r_2) >\frac{(r_1-r_2)^2}{2m} 
\ge\frac{r_1^2}{8m}. \label{I3Lpmz}
\end{align}
From  (\ref{I3g}), (\ref{I3inequality3}), (\ref{I3inequality7}), and (\ref{I3Lpmz}), it follows that
\begin{align*}
-g(r_1, r_2, z)&\le\frac{2}{m}\frac{z^2r_1^4r_2^4}{r_1r_2}\frac{17m^2}{4r_2^4}\frac{5m}{2r_2^2} 2(\frac{1}{m}+\frac{2\sqrt{\kappa^2+\nu^2}}{\kappa^2})r_1^2(\frac{8m}{r_1^2})^4  \\
&=2^{11}\cdot5\cdot17 m^6(\frac{1}{m}+\frac{2\sqrt{\kappa^2+\nu^2}}{\kappa^2})\frac{z^2}{r_1^3r_2^3}.
\end{align*}
Hence
\begin{align}
-2\int_{0}^1dz\int_{\kappa}^{\Lambda}dr_2\int_{2r_2}^{\Lambda}g(r_1, r_2, z)dr_1
\le2^7\cdot5\cdot17\frac{m^6}{\kappa^4}(\frac{1}{m}+\frac{2\sqrt{\kappa^2+\nu^2}}{\kappa^2}). \label{I3integral2}
\end{align}
Then by (\ref{I3SLambdaintegral}),  (\ref{I3integral1}), and (\ref{I3integral2}), we have
\begin{align*}
-S(\Lambda)\le\frac{1190m^4}{\kappa^2}(\frac{1}{m}+\frac{2\sqrt{\kappa^2+\nu^2}}{\kappa^2})+2^7\cdot5\cdot17\frac{m^6}{\kappa^4}(\frac{1}{m}+\frac{2\sqrt{\kappa^2+\nu^2}}{\kappa^2}).
\end{align*}
Since $S(\Lambda)$ is decreasing and bounded below, it converges as $\Lambda\to\infty$.  This fact proves the lemma.
\qed

\subsubsection{Proof of $\lim_{\Lambda\to\infty}\frac{{\rm I}_4(\Lambda)}{\log\Lambda}=0$}
We redefine $h(r_1, r_2, z)$ as
\begin{align*}
h(r_1, r_2, z)=\frac{r_1^2r_2^2}{\omega(r_1)\omega(r_2)}(\frac{r_1^2}{F(r_1)^2}+\frac{r_2^2}{F(r_2)^2})(\frac{1}{F(r_1)}+\frac{1}{F(r_2)})\frac{1}{L(r_1,r_2, z)^2}.
\end{align*}
Then we have
$
{\rm I}_4(\Lambda)=\frac{2\pi^2}{(2\pi)^6}\int_{-1}^1dz\int_{\kappa}^{\Lambda}dr_1\int_{\kappa}^{\Lambda}h(r_1, r_2, z)dr_2$. 
We define $J(\Lambda)$ as
\begin{align*}
J(\Lambda)= \int_{-1}^1dz\int_{\kappa}^{\Lambda}dr_1\int_{\kappa}^{\Lambda}h(r_1, r_2, z)dr_2. 
\end{align*}

{\bf (Step 1)}
We define $S(\Lambda, z)$ as
\begin{align*}
S(\Lambda, z)=\int_{\kappa}^{\Lambda}dr_1\int_{\kappa}^{\Lambda}h(r_1, r_2, z)dr_2. 
\end{align*}
Our first task is to prove that $\displaystyle \lim_{\Lambda\to\infty}S(\Lambda, z)$ exists for all $ z\in {\rm I}=[-1, 1].$ Since $h(r_1, r_2, z)>0$, $S(\Lambda, z)$ is increasing in $\Lambda$. Let 
\begin{align*}
&S_1(\Lambda, z)=2\int_{\kappa}^{\Lambda}dr_2\int_{r_2}^{2r_2}h(r_1, r_2, z)dr_1,\\
&S_2(\Lambda, z)=2\int_{\kappa}^{\Lambda}dr_2\int_{2r_2}^{\Lambda}h(r_1, r_2, z)dr_1.
\end{align*}
Then 
\begin{align}
S(\Lambda, z)&=\int_{\kappa}^{\Lambda}dr_2\int_{r_2}^{\Lambda}h(r_1, r_2, z)dr_1+\int_{\kappa}^{\Lambda}dr_1\int_{r_1}^{\Lambda}h(r_1, r_2, z)dr_2 \notag \\
&=2\int_{\kappa}^{\Lambda}dr_2\int_{r_2}^{\Lambda}h(r_1, r_2, z)dr_1 \notag\\
&=2\int_{\kappa}^{\Lambda}dr_2\int_{r_2}^{2r_2}h(r_1,r_2, z)dr_1+2\int_{\kappa}^{\Lambda}dr_2\int_{2r_2}^{\Lambda}h(r_1, r_2, z)dr_1 \notag \\
&=S_1(\Lambda, z)+S_2(\Lambda, z) \label{I4SLambdasum}
\end{align}
 holds.  
 We have 
\begin{align}
\frac{1}{F(r_1)}+\frac{1}{F(r_2)}<2m(\frac{1}{r_1^2}+\frac{1}{r_2^2}), \quad 
\frac{r_1^2}{F(r_1)^2}+\frac{r_2^2}{F(r_2)^2}<4m^2(\frac{1}{r_1^2}+\frac{1}{r_2^2}) \label{Fsquareinversesem1}
\end{align}
and 
\begin{align}
r_1<L(r_1, r_2, z). \label{Lzlargerthanr1}
\end{align} 
Let $r_2\le r_1 \le 2r_2$.  Since $1/r_2^2 \le 4/r_1^2$, it holds that
\begin{align}
\frac{1}{r_1^2}+\frac{1}{r_2^2}\le 5/r_1^2.  \label{rinversesum}
\end{align} 
Then from  (\ref{Fsquareinversesem1}) and (\ref{rinversesum}), 
it follows that 
\begin{align}
\frac{1}{F(r_1)}+\frac{1}{F(r_2)}<\frac{10m}{r_1^2},\quad 
\frac{r_1^2}{F(r_1)^2}+\frac{r_2^2}{F(r_2)^2}<\frac{20m^2}{r_1^2}.  \label{Fsquareinversesum2}
\end{align}
Hence from (\ref{Lzlargerthanr1}) and (\ref{Fsquareinversesum2}), it follows that
$
h(r_1, r_2, z)<\frac{200m^3r_2}{r_1^5}$. 
Therefore we have 
\begin{align}
S_1(\Lambda, z)&<400m^3\int_{\kappa}^{\Lambda}r_2dr_2\int_{r_2}^{2r_2}\frac{dr_1}{r_1^5} =\frac{375m^3}{8}(\frac{1}{\kappa^2}-\frac{1}{\Lambda^2}) <\frac{375m^3}{8\kappa^2}. \label{eq:S1Lambda}
\end{align}
Let $2r_2\le r_1$. Since $r_1/2\le r_1-r_2$, we have
\begin{align}
\frac{r_1^2}{8m}\le \frac{(r_1-r_2)^2}{2m}<L(r_1, r_2, z). \label{Lzlargerthan2}
\end{align}
Since $r_2^2\le r_1^2,$ we have
\begin{align}
\frac{1}{F(r_1)}+\frac{1}{F(r_2)}<\frac{4m}{r_2^2}, \quad
\frac{r_1^2}{F(r_1)^2}+\frac{r_2^2}{F(r_2)^2}<\frac{8m^2}{r_2^2}. \label{Fsquareinversesum3}
\end{align}
Hence from (\ref{Lzlargerthan2}) and (\ref{Fsquareinversesum3}), it follows that
$
h(r_1, r_2, z)<\frac{2048m^5}{r_1^3r_2^3}$. 
Therefore  we have
\begin{align}
S_2(\Lambda, z)<4096m^5\int_{\kappa}^{\Lambda}\frac{dr_2}
{r_2^3}\int_{2r_2}^{\Lambda}\frac{dr_1}{r_1^3}
<512m^5\int_{\kappa}^{\Lambda}\frac{dr_2}{r_2^5}
<\frac{128m^5}{\kappa^4}. \label{eq:S2Lambda}
\end{align}
From (\ref{I4SLambdasum}),  (\ref{eq:S1Lambda}), and (\ref{eq:S2Lambda}), it follows that
$$ S(\Lambda, z)<\frac{375m^3}{8\kappa^2}+\frac{128m^5}{\kappa^4}. $$ 
Since $S(\Lambda, z)$ is increasing in $\Lambda$ and bounded above for all $z\in {\rm I}$, it converges as $\Lambda$ goes to infinity.

{\bf (Step 2)}
Our second task is to prove that $J(\Lambda)$ converges when $\Lambda$ goes to infinity.
Let $M(r_1, r_2)$ be
$$M(r_1, r_2)=\frac{r_1^2r_2^2}{\omega(r_1)\omega(r_2)}(\frac{r_1^2}{F(r_1)^2}+\frac{r_2^2}{F(r_2)^2})(\frac{1}{F(r_1)}+\frac{1}{F(r_2)})\frac{1}{L(r_1, r_2, -1)^2}.$$
$|h(r_1, r_2, z)|\le M(r_1, r_2)$ holds for all $(r_1, r_2, z)\in [\kappa, \Lambda]^2\times {\rm I}$,  and by (Step 1), there exists
$$\lim_{\Lambda\to\infty}\int_{\kappa}^{\Lambda} \int_{\kappa}^{\Lambda}M(r_1, r_2)dr_1dr_2.$$
Since 
$$M_{\Lambda}=\int_{\kappa}^{\Lambda}\int_{\kappa}^{\Lambda}M(r_1, r_2)dr_1dr_2 \to M_{\infty}=\lim_{\Lambda\to\infty}\int_{\kappa}^{\Lambda} \int_{\kappa}^{\Lambda}M(r_1, r_2)dr_1dr_2,$$
from Cauchy convergence condition, for any $\epsilon >0,$ there exists $\Lambda_0 \in [\kappa, \infty)$ such that if $\Lambda_0<\Lambda_1\le\Lambda_2$, $|M_{\Lambda_2}-M_{\Lambda_1}|<\epsilon$. Then for $\Lambda_0<\Lambda_1\le\Lambda_2$ and all $z \in {\rm I}$, 
\begin{align*}
&|S(\Lambda_2, z)-S(\Lambda_1, z)| \\
&=|\int_{\kappa}^{\Lambda_1}dr_1\int_{\Lambda_1}^{\Lambda_2}h(r_1, r_2, z)dr_2+\int_{\kappa}^{\Lambda_1}dr_2\int_{\Lambda_1}^{\Lambda_2}h(r_1, r_2, z)dr_1+\int_{\Lambda_1}^{\Lambda_2}\int_{\Lambda_1}^{\Lambda_2}h(r_1, r_2, z)dr_1dr_2| \\
 &\le\int_{\kappa}^{\Lambda_1}dr_1\int_{\Lambda_1}^{\Lambda_2}M(r_1, r_2)dr_2+\int_{\kappa}^{\Lambda_1}dr_2\int_{\Lambda_1}^{\Lambda_2}M(r_1, r_2)dr_1+\int_{\Lambda_1}^{\Lambda_2}\int_{\Lambda_1}^{\Lambda_2}M(r_1, r_2)dr_1dr_2 \\
&=|M_{\Lambda_2}-M_{\Lambda_1}|
<\epsilon .
\end{align*}
Therefore
$\sup_{z\in I}|S(\Lambda_2, z)-S(\Lambda_1, z)|\le|M_{\Lambda_2}-M_{\Lambda_1}|<\epsilon$ holds. Since family of functions  $(S(\Lambda, \cdot))_{\Lambda\in[\kappa, \infty)}$ on ${\rm I}$ satisfies uniform Cauchy conditions, it converges uniformly on ${\rm I}$. Since $[\kappa, \Lambda]^2$ is a Jordan measurable bounded closed set of  ${\mathbb R}^2,$ the function $S(\Lambda, z)$  is continuous on ${\rm I}.$ Hence  
$$ S(\infty, z)=\lim_{\Lambda\to\infty}\int_{\kappa}^{\Lambda} \int_{\kappa}^{\Lambda}h(r_1, r_2, z)dr_1dr_2$$ 
is continuous on ${\rm I}$. Since both $S(\Lambda, z)$ and $S(\infty, z)$ are integrable on Jordan measurable set ${\rm I}$, by uniform convergence theorem, we have
$$\lim_{\Lambda\to\infty}\int_{-1}^1S(\Lambda, z)dz=\int_{-1}^1S(\infty, z)dz.$$
It implies that  $J(\Lambda)$ converges as $\Lambda\to\infty$.
\qed

\subsubsection{Proof of $\lim_{\Lambda\to\infty}\frac{{\rm I}_5(\Lambda)}{\log\Lambda}=0$}
We redefine $h(r_1, r_2, z),$ $g(r_1, r_2, z),$ and $S(\Lambda)$ as
\begin{align*}
&h(r_1, r_2, z)=\frac{zr_1^3r_2^3}{\omega(r_1)\omega(r_2)F(r_1)^2F(r_2)^2L(r_1, r_2, z)},\\
&g(r_1, r_2, z)=h(r_1, r_2, z)+h(r_1, r_2, -z),\\
&
S(\Lambda)=\int_{-1}^1dz\int_{\kappa}^{\Lambda}dr_1\int_{\kappa}^{\Lambda}h(r_1, r_2, z)dr_2.
\end{align*}
Then
$
{\rm I}_5(\Lambda)=\frac{2\pi^2}{(2\pi)^6}S(\Lambda)$. 
We have 
$$S(\Lambda)=\int_{0}^{1}dz\int_{\kappa}^{\Lambda}dr_2\int_{\kappa}^{\Lambda}g(r_1, r_2, z)dr_1$$
in the same way as (\ref{SLambdadivide}). Since
$$g(r_1, r_2, z)=-\frac{2z^2r_1^4r_2^4}{m^2\omega(r_1)\omega(r_2)F(r_1)^2F(r_2)^2L(r_1, r_2, z)L(r_1, r_2, -z)}\le0,$$
$S(\Lambda)$ is decreasing in $\Lambda$. Since $r_1<L(r_1, r_2, z)$, we have
$$ \frac{r_1^4}{\omega(r_1)F(r_1)^2L(r_1, r_2, z)}<\frac{4m^2}{r_1^2}.$$ 
Similarly, we have
$$ \frac{r_2^4}{\omega(r_2)F(r_2)^2L(r_1, r_2, -z)}<\frac{4m^2}{r_2^2}.$$ 
Hence
\begin{align*}
-S(\Lambda)&=\frac{2}{m^2}\int_{0}^{1}dz z^2\int_{\kappa}^{\Lambda}dr_2\frac{r_2^4}{\omega(r_2)F(r_2)^2L(r_1, r_2, -z)}\int_{\kappa}^{\Lambda}\frac{r_1^4}{\omega(r_1)F(r_1)^2L(r_1, r_2, z)}dr_1 \notag \\
&<\frac{2}{3m^2}\int_{\kappa}^{\Lambda}dr_2\frac{4m^2}{r_2^2}\int_{\kappa}^{\Lambda}\frac{4m^2}{r_1^2}dr_1 
<\frac{32m^2}{3\kappa^2}.
\end{align*}
Since $S(\Lambda)$ is decreasing and bounded below, it converges as $\Lambda\to\infty$.
\qed

\subsubsection{Proof of $\lim_{\Lambda\to\infty}\frac{{\rm I}_6(\Lambda)}{\log\Lambda}=0$}
We redefine $h(r_1,r_2, z)$, $J(\Lambda),$ $S(\Lambda, z),$ $S_1(\Lambda, z),$ and $S_2(\Lambda, z)$ as
\begin{align*}
&h(r_1, r_2, z)=\frac{r_1^2r_2^2}{\omega(r_1)\omega(r_2)}(\frac{1}{F(r_1)}+\frac{1}{F(r_2)})^2\frac{r_1^2+r_2^2}{L(r_1, r_2, z)^3},\\
&J(\Lambda)=\int_{-1}^1dz\int_{\kappa}^{\Lambda}dr_2\int_{\kappa}^{\Lambda}h(r_1, r_2, z)dr_1, \\
&S(\Lambda, z)= \int_{\kappa}^{\Lambda}dr_1\int_{\kappa}^{\Lambda}h(r_1, r_2, z)dr_2, \\
& S_1(\Lambda, z)=2\int_{\kappa}^{\Lambda}dr_2\int_{r_2}^{2r_2}h(r_1, r_2, z)dr_1,\\
& S_2(\Lambda, z)=2\int_{\kappa}^{\Lambda}dr_2\int_{2r_2}^{\Lambda}h(r_1, r_2, z)dr_1.
\end{align*}
We have
$
{\rm I}_6(\Lambda)=\frac{\pi^2}{(2\pi)^6}J(\Lambda)$.

{\bf (Step 1)}
Our first task is to prove that $\displaystyle\lim_{\Lambda\to\infty}S(\Lambda, z)$ exists for all  $z\in {\rm I}.$ Since $h(r_1, r_2, z)>0,$ $S(\Lambda, z)$ is increasing in $\Lambda.$ 
We have
\begin{align}
S(\Lambda, z)=S_1(\Lambda, z)+S_2(\Lambda, z)  \label{sumofS1S2}
\end{align}
in the same way as (\ref{I4SLambdasum}).  When $r_2\le r_1$, it holds that
\begin{align*}
r_1^2+r_2^2 \le 2r_1^2,\quad 
\frac{1}{F(r_1)}+\frac{1}{F(r_2)}<\frac{4m}{r_2^2}. 
\end{align*}
When $r_2\le r_1 \le 2r_2$, it also holds that 
\begin{align*}
\frac{1}{F(r_1)}+\frac{1}{F(r_2)}<\frac{10m}{r_1^2}. 
\end{align*}
Then we have
\begin{align}
\int_{r_2}^{2r_2}h(r_1, r_2, z)dr_1&<\int_{r_2}^{2r_2}\frac{r_1^2r_2^2}{r_1r_2}  (\frac{10m}{r_1^2})^2\frac{2r_1^2}{r_1^3} dr_1 
=\frac{175m^2}{3r_2^2}. \notag
\end{align}
Hence
\begin{align}
S_1(\Lambda, z)&<\frac{350m^2}{3}\int_{\kappa}^{\Lambda}\frac{dr_2}{r_2^2} 
<\frac{350m^2}{3\kappa}. \label{S1Lambdazsmaller}
\end{align}
Let $2r_2\le r_1$. Since $ r_1/2 \le r_1-r_2,$ we have
$$ \frac{r_1^2}{8m}\le \frac{(r_1-r_2)^2}{2m}<L(r_1, r_2, z).$$ 
Then
\begin{align*}
\int_{2r_2}^{\Lambda}h(r_1, r_2, z)dr_1&<\int_{2r_2}^{\Lambda}\frac{r_1^2r_2^2}{r_1r_2}(\frac{4m}{r_2^2})^2(\frac{8m}{r_1^2})^3 2r_1^2dr_1 =\frac{16384m^5}{r_2^3}\int_{2r_2}^{\Lambda}\frac{dr_1}{r_1^3} 
<\frac{2048m^5}{r_2^5}.
\end{align*}
Therefore
\begin{align}
S_2(\Lambda, z)&<4096m^5\int_{\kappa}^{\Lambda}\frac{dr_2}{r_2^5} 
< \frac{1024m^5}{\kappa^4}. \label{S2Lambdazsmaller}
\end{align}
From (\ref{sumofS1S2}), (\ref{S1Lambdazsmaller}), and (\ref{S2Lambdazsmaller}), it follows that
$S(\Lambda, z)<\frac{350m^2}{3\kappa}+\frac{1024m^5}{\kappa^4}$. 
Since $S(\Lambda, z)$ is  increasing in $\Lambda$ and bounded above, it converges as $\Lambda$ goes to infinity.

{\bf (Step 2)}
Our second task is to prove $J(\Lambda)$ converges as $\Lambda$ goes to infinity. This step is the same as that of $\lim_{\Lambda\to\infty}\frac{{\rm I}_4(\Lambda)}{\log\Lambda}=0$.
\qed

\subsubsection{Proof of $\lim_{\Lambda\to\infty}\frac{{\rm I}_7(\Lambda)}{\log\Lambda}=0$}
We redefine $h(r_1, r_2, z),$ $g(r_1, r_2, z),$ and $S(\Lambda)$ as
\begin{align*}
&h(r_1, r_2, z)=\frac{zr_1^3r_2^3}{\omega(r_1)\omega(r_2)}(\frac{1}{F(r_1)}+\frac{1}{F(r_2)})^2\frac{1}{L(r_1, r_2, z)^3},\\
&g(r_1, r_2, z)=h(r_1, r_2, z)+h(r_1, r_2, -z),\\
&S(\Lambda)=\int_{-1}^1dz\int_{\kappa}^{\Lambda}dr_1\int_{\kappa}^{\Lambda}h(r_1, r_2, z)dr_2.
\end{align*}
Then we have
$
{\rm I}_7(\Lambda)=\frac{2\pi^2}{(2\pi)^6}S(\Lambda)$, 
and
$$S(\Lambda)=\int_0^1dz\int_{\kappa}^{\Lambda}dr_1\int_{\kappa}^{\Lambda}g(r_1, r_2, z)dr_2$$
in the same way as $\lim_{\Lambda\to\infty}\frac{{\rm I}_3(\Lambda)}{\log\Lambda}=0$. We define $g_1(r_1, r_2, z)$ and $g_2(r_1, r_2, z)$ as
\begin{align*}
&g_1(r_1, r_2, z)=-\frac{6z^2r_1^4r_2^4}{m\omega(r_1)\omega(r_2)}(\frac{1}{F(r_1)}+\frac{1}{F(r_2)})^2
\frac{
(\frac{r_1^2+r_2^2}{2m}+\omega(r_1)+\omega(r_2))^2}
{L(r_1, r_2, z)^3L(r_1, r_2, -z)^3} ,\\
&g_2(r_1, r_2, z)=-\frac{2z^4r_1^6r_2^6}{m^3\omega(r_1)\omega(r_2)}(\frac{1}{F(r_1)}+\frac{1}{F(r_2)})^2
\frac{1}{L(r_1, r_2, z)^3L(r_1, r_2, -z)^3}.  
\end{align*}
Then we have
$ g(r_1, r_2, z)=g_1(r_1, r_2, z)+g_2(r_1, r_2, z)$. 
We redefine $S_1(\Lambda)$ and $S_2(\Lambda)$ by
\begin{align*}
&S_1(\Lambda)=
\int_0^1dz\int_{\kappa}^{\Lambda}dr_1\int_{\kappa}^{\Lambda}g_1(r_1, r_2, z)dr_2,\\
&S_2(\Lambda)=
\int_0^1dz\int_{\kappa}^{\Lambda}dr_1\int_{\kappa}^{\Lambda}g_2(r_1, r_2, z)dr_2.
\end{align*}
Then
\begin{align*}
S(\Lambda)=S_1(\Lambda)+S_2(\Lambda).
\end{align*}
Since $g_1(r_1, r_2, z)\le0$, $S_1(\Lambda)$ is decreasing in $\Lambda$. We divide $S_1(\Lambda)$ in the following way:
\begin{align*}
S_1(\Lambda)&= \int_0^1dz\int_{\kappa}^{\Lambda}dr_1\int_{r_1}^{\Lambda}g_1(r_1, r_2, z)dr_2 +\int_0^1dz\int_{\kappa}^{\Lambda}dr_2\int_{r_2}^{\Lambda}g_1(r_1, r_2, z)dr_1 \notag \\
&= \int_0^1dz\int_{\kappa}^{\Lambda}dr_1\int_{r_1}^{\Lambda}g_1(r_1, r_2, z)dr_2+ \int_0^1dz\int_{\kappa}^{\Lambda}dr_1\int_{r_1}^{\Lambda}g_1(r_2, r_1, z)dr_2 \notag \\
&=2\int_0^1dz\int_{\kappa}^{\Lambda}dr_1\int_{r_1}^{\Lambda}g_1(r_1, r_2, z)dr_2. \end{align*}
Let $\kappa\le r_1 \le r_2$.  Then we have
\begin{align*}
\frac{r_1^2+r_2^2}{2m}+\omega(r_1)+\omega(r_2)\le (\frac{1}{m}+\frac{2\sqrt{\kappa^2+\nu^2}}{\kappa^2})r_2^2 
\end{align*} 
in the same way as (\ref{I3inequality3}).  Let $r_1\le r_2 \le 2r_1$. Then we also have
\begin{align*}
\frac{1}{F(r_1)}+\frac{1}{F(r_2)}<\frac{10m}{r_2^2}.  
\end{align*}
Therefore 
\begin{align*}
-g_1(r_1, r_2, z)\le\frac{6}{m}\frac{z^2r_1^4r_2^4}{r_1r_2}(\frac{10m}{r_2^2})^2 \{(\frac{1}{m}+\frac{2\sqrt{\kappa^2+\nu^2}}{\kappa^2})r_2^2\}^2(\frac{2m}{r_2^2})^3\frac{1}{r_2^3} =4800m^4(\frac{1}{m}+\frac{2\sqrt{\kappa^2+\nu^2}}{\kappa^2})^2\frac{z^2r_1^3}{r_2^6},
\end{align*}
and
\begin{align*}
-\int_{r_1}^{2r_1}g_1(r_1, r_2, z)dr_2&\le4800m^4(\frac{1}{m}+\frac{2\sqrt{\kappa^2+\nu^2}}{\kappa^2})^2z^2r_1^3\int_{r_1}^{2r_1}\frac{dr_2}{r_2^6} =930m^4(\frac{1}{m}+\frac{2\sqrt{\kappa^2+\nu^2}}{\kappa^2})^2\frac{z^2}{r_1^2}. 
\end{align*}
Hence
\begin{align*}
-2\int_0^1\!\!\! 
dz\int_{\kappa}^{\Lambda}
\!\!\!
dr_1\int_{r_1}^{2r_1}\!\!\! \!\!\! g_1(r_1, r_2, z)
dr_2&\le1860m^4(\frac{1}{m}+\frac{2\sqrt{\kappa^2+\nu^2}}{\kappa^2})^2\int_0^1
\!\!\! z^2dz\int_{\kappa}^{\Lambda}\frac{dr_2}{r_2^2} <\frac{620m^4}{\kappa}(\frac{1}{m}+\frac{2\sqrt{\kappa^2+\nu^2}}{\kappa^2})^2. 
\end{align*} 
Let $2r_1\le r_2$. Then we have
\begin{align*}
\frac{1}{F(r_1)}+\frac{1}{F(r_2)}<\frac{5m}{2r_1^2}.
\end{align*}
In addition, since $r_2/2 \le r_2-r_1$, we can see that
\begin{align*} 
\frac{r_2^2}{8m}\le\frac{(r_2-r_1)^2}{2m}<L(r_1, r_2, -z).
\end{align*}
Therefore  
\begin{align*}
-g_1(r_1, r_2, z)&\le\frac{6}{m}\frac{z^2r_1^4r_2^4}{r_1r_2}(\frac{5m^2}{2r_1^2})^2\{(\frac{1}{m}+\frac{2\sqrt{\kappa^2+\nu^2}}{\kappa^2})r_2^2\}^2(\frac{2m}{r_2^2})^3(\frac{8m}{r_2^2})^3 \notag\\
&=153600m^9(\frac{1}{m}+\frac{2\sqrt{\kappa^2+\nu^2}}{\kappa^2})^2\frac{z^2}{r_1r_2^5}.
\end{align*}
Then we have
\begin{align*}
-\int_{2r_1}^{\Lambda}g_1(r_1, r_2, z)dr_2&\le\frac{153600m^9z^2}{r_1}(\frac{1}{m}+\frac{2\sqrt{\kappa^2+\nu^2}}{\kappa^2})^2\int_{2r_1}^{\Lambda}\frac{dr_2}{r_2^5}\le2400m^9(\frac{1}{m}+\frac{2\sqrt{\kappa^2+\nu^2}}{\kappa^2})^2\frac{z^2}{r_1^5}.
\end{align*}
Hence 
\begin{align*}
-2\int_0^1dz\int_{\kappa}^{\Lambda}dr_1\int_{2r_1}^{\Lambda}
g_1(r_1, r_2, z)dr_2&\le4800m^9(\frac{1}{m}+\frac{2\sqrt{\kappa^2+\nu^2}}{\kappa^2})^2\int_0^1z^2dz\int_{\kappa}^{\Lambda}\frac{dr_1}{r_1^5} \notag \\
&<\frac{400m^9}{\kappa^4}(\frac{1}{m}+\frac{2\sqrt{\kappa^2+\nu^2}}{\kappa^2})^2. 
\end{align*}
Then we have
$$-S_1(\Lambda)<\frac{620m^4}{\kappa}(\frac{1}{m}+\frac{2\sqrt{\kappa^2+\nu^2}}{\kappa^2})^2+\frac{400m^9}{\kappa^4}(\frac{1}{m}+\frac{2\sqrt{\kappa^2+\nu^2}}{\kappa^2})^2.$$
Since $S_1(\Lambda)$ is decreasing and bounded below, it converges as $\Lambda\to\infty$. Since $g_2(r_1, r_2, z)\le0,$ $S_2(\Lambda)$ is also decreasing in $\Lambda$. Let $r_1 \le r_2$. Then
\begin{align*}
\frac{1}{F(r_1)}+\frac{1}{F(r_2)}<\frac{4m}{r_1^2}. 
\end{align*}
Therefore 
\begin{align*}
-g_2(r_1, r_2, z)&\le\frac{2}{m^3}\frac{z^4r_1^6r_2^6}{r_1r_2}(\frac{4m}{r_1^2})^2\frac{8m^3}{r_2^6}\frac{1}{r_2^3} 
=\frac{256m^2z^4 r_1}{r_2^4}.
\end{align*}
Then 
\begin{align*}
-\int_{r_1}^{\Lambda}g_2(r_1, r_2, z)dr_2&\le256m^2z^4r_1\int_{r_1}^{\Lambda}\frac{dr_2}{r_2^4} 
\le\frac{256m^2z^4}{3r_1^2}.
\end{align*}
Hence
\begin{align*}
-S_2(\Lambda)& \le\frac{512m^2}{3}\int_0^1z^4dz \int_{\kappa}^{\Lambda}\frac{dr_1}{r_1^2} 
<\frac{512m^2}{15\kappa}.
\end{align*}
Since $S_2(\Lambda)$ is  decreasing in $\Lambda$ and bounded below, it converges. Since both    $S_1(\Lambda)$ and $S_2(\Lambda)$ converge, $S(\Lambda)$ converges.
\qed

\subsubsection{Proof of $\lim_{\Lambda\to\infty}\frac{{\rm I}_8(\Lambda)}{\log\Lambda}=0$}
We redefine $h(r_1, r_2, z),$ $g(r_1, r_2, z),$ $g_1(r_1, r_2, z),$ $g_2(r_1, r_2, z),$ $S(\Lambda),$ $S_1(\Lambda),$ and $S_2(\Lambda)$ as
\begin{align}
&h(r_1, r_2, z)=\frac{zr_1^3r_2^3}{\omega(r_1)\omega(r_2)}(\frac{1}{F(r_1)}+\frac{1}{F(r_2)})\frac{1}{L(r_1, r_2, z)^4},\nonumber\\
&g(r_1, r_2, z)=h(r_1, r_2, z)+h(r_1, r_2, -z),\nonumber\\
&g_1(r_1, r_2, z)=-\frac{2z^2r_1^4r_2^4}{m\omega(r_1)\omega(r_2)}(\frac{1}{F(r_1)}+\frac{1}{F(r_2)})\frac{1}{L(r_1, r_2, z)L(r_1, r_2, -z)^4}, \label{I8g1}\\
&g_2(r_1, r_2, z)=-\frac{2z^2r_1^4r_2^4}{m\omega(r_1)\omega(r_2)}(\frac{1}{F(r_1)}+\frac{1}{F(r_2)})G(r_1,r_2, z), \label{I8g2}\\
&S(\Lambda)=\int_{-1}^1dz\int_{\kappa}^{\Lambda}dr_2\int_{\kappa}^{\Lambda}h(r_1, r_2, z)dr_1,\nonumber\\
&S_1(\Lambda)=\int_{\kappa}^{\Lambda}dr_2\int_{r_2}^{\Lambda}dr_1
\int_0^{1-\frac{1}{r_1^{1/4}r_2^{1/2}}}g_1(r_1, r_2, z)dz,\nonumber\\
&S_2(\Lambda)=\int_{\kappa}^{\Lambda}dr_2\int_{r_2}^{\Lambda}dr_1
\int_{1-\frac{1}{r_1^{1/4}r_2^{1/2}}}^1g_1(r_1, r_2, z)dz,\nonumber
\end{align}
where 
\begin{align*}
G(r_1,r_2, z)=\frac{1}{L(r_1, r_2, z)^2L(r_1, r_2, -z)^3}+\frac{1}
{L(r_1, r_2, z)^3L(r_1, r_2, -z)^2}+\frac{1}{L(r_1, r_2, z)^4L(r_1, r_2, -z)}. 
\end{align*}
Furthermore, we define $S_3(\Lambda)$ as
\begin{align*}
S_3(\Lambda)=\int_0^1dz\int_{\kappa}^{\Lambda}dr_2
\int_{r_2}^{\Lambda}g_2(r_1, r_2, z)dr_1.
\end{align*}
Then we have
$
{\rm I}_8(\Lambda)=\frac{2\pi^2}{(2\pi)^6}S(\Lambda)$, 
and
$$S(\Lambda)=2 \int_0^1dz \int_{\kappa}^{\Lambda}dr_2\int_{r_2}^{\Lambda}g(r_1, r_2, z)dr_1$$
in the same way as the proof of $\lim_{\Lambda\to\infty}\frac{{\rm I}_3(\Lambda)}{\log\Lambda}=0$. 
 Since 
$g(r_1,r_2, z)=g_1(r_1,r_2, z)+g_2(r_1,r_2, z)$, it holds that
\begin{align}
S(\Lambda)=2S_1(\Lambda)+2S_2(\Lambda)+2S_3(\Lambda). \label{eq:sumofSi}
\end{align}
Since $g_1(r_1, r_2, z)\le0$ and $g_2(r_1, r_2, z)\le0,$ $S_i(\Lambda)$ $(i=1, 2, 3)$ are  decreasing in $\Lambda$. Let $r_2 \le r_1$. Then
\begin{align}
\frac{1}{F(r_1)}+\frac{1}{F(r_2)} < \frac{4m}{r_2^2}. \label{eq:sumofFinverse}
\end{align}
Let $0\le z \le 1-\frac{1}{r_1^{1/4}r_2^{1/2}}$. Then we have
\begin{align}
\frac{1}{(1-z)^4}\le r_1r_2^2.   \label{eq:lemma195}
\end{align}  
We have
\begin{align}
 L(r_1, r_2, -z)=\frac{(r_1-r_2)^2+2r_1r_2(1-z)}{2m}+\omega(r_1)+\omega(r_2)>\frac{r_1r_2(1-z)}{m}. \label{Lminuszlarger2}
\end{align} 
Using (\ref{eq:lemma195}) and (\ref{Lminuszlarger2}), we have
\begin{align}
\frac{1}{L(r_1, r_2, -z)^4}<\frac{m^4}{r_1^4r_2^4(1-z)^4}\le\frac{m^4}{r_1^3r_2^2}. \label{eq:lemma196}
\end{align}
From (\ref{I8g1}), (\ref{eq:sumofFinverse}), and (\ref{eq:lemma196}), it follows that 
\begin{align*}
-g_1(r_1, r_2, z)\le\frac{2}{m}\frac{z^2r_1^4r_2^4}{r_1r_2}\frac{4m}{r_2^2}\frac{2m}{r_1^2}\frac{m^4}{r_1^3r_2^2}  
=\frac{16m^5}{r_1^2r_2}.
\end{align*}
Hence we have
\begin{align*}
 -S_1(\Lambda)&\le16m^5\int_{\kappa}^{\Lambda}\frac{dr_2}{r_2}\int_{r_2}^{\Lambda}\frac{1}{r_1^2}(1-\frac{1}{r_1^{1/4}r_2^{1/2}})dr_1 
<16m^5\int_{\kappa}^{\Lambda}\frac{dr_2}{r_2}\int_{r_2}^{\Lambda}\frac{dr_1}{r_1^2} <\frac{16m^5}{\kappa}.
\end{align*}
Since $S_1(\Lambda)$  is decreasing in $\Lambda$ and bounded below, it converges. When $r_2 \le r_1$ and $1-\frac{1}{r_1^{1/4}r_2^{1/2}}\le z \le  1$, 
from (\ref{I8g1}) and (\ref{eq:sumofFinverse}), it holds that 
\begin{align*}
-g_1(r_1, r_2, z)<\frac{2}{m}\frac{r_1^4r_2^4}{r_1r_2}\frac{4m}{r_2^2}\frac{2m}{r_1^2}\frac{1}{r_1^4} 
=\frac{16m r_2}{r_1^3}.
\end{align*}
Hence we have
\begin{align*}
&-S_2(\Lambda)<16m\int_{\kappa}^{\Lambda}dr_2\int_{r_2}^{\Lambda}dr_1\int_{1-\frac{1}{r_1^{1/4}r_2^{1/2}}}^1\frac{r_2}{r_1^3}dz 
=16m\int_{\kappa}^{\Lambda}dr_2\int_{r_2}^{\Lambda}\frac{r_2}{r_1^3} r_1^{-1/4}r_2^{-1/2}dr_1 \\
&=16m\int_{\kappa}^{\Lambda}dr_2 r_2^{1/2}\int_{r_2}^{\Lambda}r_1^{-13/4}dr_1 
=\frac{64m}{9}\int_{\kappa}^{\Lambda}r_2^{1/2}(r_2^{-9/4}-\Lambda^{-9/4})dr_2 
<\frac{64m}{9}\int_{\kappa}^{\Lambda}r_2^{-7/4}dr_2 \\
&=\frac{256m}{27}(\kappa^{-3/4}-\Lambda^{-3/4})  
<\frac{256m}{27\kappa^{3/4}}.
\end{align*}
Since $S_2(\Lambda)$  is decreasing in $\Lambda$ and bounded below, it converges.
We have
\begin{align}
G(r_1, r_2, z)< \frac{4m^2}{r_1^7}+\frac{8m^3}{r_1^8}+\frac{16m^4}{r_1^9}.  \label{eq:Genequality}
\end{align}
From (\ref{I8g2}), (\ref{eq:sumofFinverse}), and (\ref{eq:Genequality}), we have
\begin{align*}
-g_2(r_1, r_2, z)&\le\frac{2}{m}\frac{z^2r_1^4r_2^4}{r_1r_2}\frac{4m}{r_2^2}(\frac{4m^2}{r_1^7}+\frac{8m^3}{r_1^8}+\frac{16m^4}{r_1^9} )  \\
&<8r_1^3r_2(\frac{4m^2}{r_1^7}+\frac{8m^3}{r_1^8}+\frac{16m^4}{r_1^9})  =\frac{32m^2r_2}{r_1^4}+\frac{64m^3r_2}{r_1^5}+\frac{128m^4 r_2}{r_1^6}.
\end{align*}
Hence 
\begin{align*}
-S_3(\Lambda) &<32m^2\int_{\kappa}^{\Lambda}dr_2r_2\int_{r_2}^{\Lambda}\frac{dr_1}{r_1^4}+64m^3\int_{\kappa}^{\Lambda}dr_2 r_2\int_{r_2}^{\Lambda}\frac{dr_1}{r_1^5}+128m^4\int_{\kappa}^{\Lambda}dr_2 r_2\int_{r_2}^{\Lambda}\frac{dr_1}{r_1^6}  \notag \\
&<\frac{32m^2}{3}\int_{\kappa}^{\Lambda}\frac{1}{r_2^2}dr_2+16m^3\int_{\kappa}^{\Lambda}\frac{1}{r_2^3}dr_2+\frac{128m^4}{5}\int_{\kappa}^{\Lambda}\frac{1}{r_2^4}dr_2 <\frac{32m^2}{3\kappa}+\frac{8m^3}{\kappa^2}+\frac{128m^4}{15\kappa^3}.
\end{align*}
Since $S_3(\Lambda)$  is decreasing in $\Lambda$ and bounded below, it converges. Since $S_i(\Lambda)  (i=1, 2, 3)$ converge, ${\rm I}_8(\Lambda)$ converges by  (\ref{eq:sumofSi}).
\qed

\section{Concluding remarks}
(1) 
The Nelson model is defined as the self-adjoint operator 
\begin{align}\label{nelson}
H_V=(-\frac{1}{2}\Delta +V)\otimes 1+1\otimes H_{\rm f}+\alpha \int_{{\mathbb R}^3}^\oplus \phi(x)dx,
\end{align}
acting in the Hilbert space 
$L^2({\mathbb R})\otimes {\mathcal F}\cong\int_{{\mathbb R}^3}^\oplus {\mathcal F} dx$. Here $V:{\mathbb R}^3\to{\mathbb R}$ is an external potential and 
$$\phi(x)=\frac{1}{\sqrt 2}\int \left\{
a^\dagger(k) e^{-ikx}\hat{\varphi}(k)/\sqrt{\omega(k)}+
a(k) e^{ikx}\hat{\varphi}(-k)/\sqrt{\omega(k)}
\right\} dk.$$
In the case of $V=0$, $H_{V=0}$ is translation invariant and 
the relationship between $H_V$ and $H(p)$ 
is given by 
$$H_{V=0}=\int_{{\mathbb R}^3}^\oplus H(p) dp.$$
Furthermore 
the ground state energy of $H(p=0)$ coincides with that of $H_{V=0}$.

(2)
We show that 
$m_{\rm eff}(\Lambda)/m=1+ \sum_{n=1}^\infty a_n(\Lambda)\alpha^{2n} $ and 
$\lim_{\Lambda\to\infty}a_2(\Lambda)=\pm\infty$.  
It is also expected that $\lim_{\Lambda\to\infty}a_n(\Lambda)$ diverges and 
the signatures are alternatively changed. Hence $\lim_{\Lambda\to\infty} m_{\rm eff}(\Lambda)/m$ may converge but it is not trivial to see it directly.

(3) The relativistic Nelson model is defined by replacing 
$-\frac{1}{2}\Delta +V$ with the semirelativistic Schr\"odinger operator 
$\sqrt{-\Delta+1}+V$ in \eqref{nelson}; that is,
$$H_V^{rel}=(\sqrt{-\Delta+1}+V)\otimes 1+1\otimes H_{\rm f}+\int_{{\mathbb R}^3}^\oplus \phi(x)dx.$$
Then it follows that 
$$H_{V=0}^{rel}=\int_{{\mathbb R}^3}^\oplus H^{rel}(p) dp,$$
where 
$H^{rel}(p)=\sqrt{(p-P_{\rm f})^2+1} + H_{\rm f}+\phi(0)$. 
Then the effective mass $m_{\rm eff}(\Lambda)$ 
of $H^{rel}(p)$ is defined in the same way as that of $H(p)$. 
We are also interested in seeing the asymptotic behavior of 
$m_{\rm eff}(\Lambda)$ as $\Lambda\to \infty$.
However 
$\sqrt{(p-P_{\rm f})^2+1}$ is a nonlocal operator and then estimates are rather complicated. 

Another interesting nonlocal model is the so-called semirelativistic Pauli-Fierz model defined by 
$$H_{V}^{PF}=\sqrt{\left(-i\nabla\otimes -\alpha \int_{{\mathbb R}^3} ^\oplus A(x) dx\right)^2+1}+V\otimes 1+1\otimes H_{\rm f},$$
where $A(x)$ is a quantized radiation field. See \cite{HiroshimaIto0} for the detail. 
Then it follows that 
$$H_{V=0}^{PF}=\int_{{\mathbb R}^3}^\oplus H^{PF}(p) dp,$$
where 
$H^{PF}(p)=\sqrt{\left(p-P_{\rm f}-\alpha A(0)\right)^2+1} + H_{\rm f}$. 
It is also interesting to investigate the asymptotic behavior of 
the effective mass of the semirelativistic Pauli-Fierz model.

\section*{Acknowledgment}
S.O. is grateful to Asao Arai for helpful comments and financial support.
{

}

\end{document}